\documentclass[11pt]{article}
\usepackage{bm}
\usepackage[margin=1in]{geometry}
\usepackage{bbm}
\usepackage{amsmath,amstext,amsthm,amssymb}
\usepackage{float}
\usepackage{graphicx}
\usepackage{mathtools}
\usepackage{wrapfig}
\usepackage{multirow}
\usepackage{cleveref}
\usepackage{thmtools}
\usepackage{thm-restate}
\usepackage{bbm}
\usepackage[bottom]{footmisc}
\usepackage[font={small,it}]{caption}

\usepackage[ruled,vlined,linesnumbered,algonl]{algorithm2e}
\newcommand{\norm}[1]{\left\lVert#1\right\rVert}
\newcommand{\sumL}{\sum\limits}

\DeclareMathOperator*{\argmin}{arg\,min}

\theoremstyle{plain}
\newtheorem{theorem}{Theorem}[section]
\newtheorem{lemma}[theorem]{Lemma}

\newtheorem{cor}[theorem]{Corollary}
\newtheorem{remark}[theorem]{Remark}
\theoremstyle{plain}
\theoremstyle{definition}
\newtheorem{defn}[theorem]{Definition}
\theoremstyle{plain}

\theoremstyle{remark}

\ifx\proof\undefined
\newenvironment{proof}[1][\protect\proofname]{\par
	\normalfont\topsep6\p@\@plus6\p@\relax
	\trivlist
	\itemindent\parindent
	\item[\hskip\labelsep\scshape #1]\ignorespaces
}{%
	\endtrivlist\@endpefalse
}
\providecommand{\proofname}{Proof}
\fi

\newcommand{\ignore}[1]{}

\newcommand{\RR}{\mathbb{R}}

\newcommand{\ZZ}{\mathbb{Z}}

\newcommand{\cA}{\mathcal{A}}

\newcommand{\cD}{\mathcal{D}}

\newcommand{\cI}{\mathcal{I}}

\newcommand{\cN}{\mathcal{N}}

\newcommand{\cP}{\mathcal{P}}

\newcommand{\cS}{\mathcal{S}}

\newcommand{\cW}{\mathcal{W}}

\newcommand{\cY}{\mathcal{Y}}

\newcommand{\eps}{\varepsilon}

\newcommand{\eat}[1]{}

\newcommand{\poly}{\operatorname{poly}}
\newcommand{\opt}{\ensuremath{\text{OPT}}\xspace}

\DeclareRobustCommand{\fairCSSx}{PC-column subset selection}
\DeclareRobustCommand{\fairCSS}{\mbox{PC-CSS}}
\DeclareRobustCommand{\fairSAp}[1]{PC-$\ell_{#1}$-subspace approximation}
\DeclareRobustCommand{\fairSA}{PC-subspace approximation }

\DeclareRobustCommand{\proj}{\mbox{proj}}

\title{Guessing Efficiently for Constrained Subspace Approximation}

\date{}
\author{
  Adita Bhaskara\thanks{University of Utah: bhaskaraaditya@gmail.com}
  \and Sepideh Mahabadi\thanks{Microsoft Research--Redmond: smahabadi@microsoft.com}
  \and Madhusudhan Reddy Pittu\thanks{Carnegie Mellon University: mpittu@andrew.cmu.edu}
  \and Ali Vakilian\thanks{Toyota Technological Institute at Chicago (TTIC): vakilian@ttic.edu}
  \and David P. Woodruff\thanks{Carnegie Mellon University: dwoodruf@andrew.cmu.edu}
}
\begin{document}
\maketitle

\begin{abstract}
In this paper we study constrained subspace approximation problem. Given a set of $n$ points $\{a_1,\ldots,a_n\}$ in $\mathbb{R}^d$, the goal of the {\em subspace approximation} problem is to find a $k$ dimensional subspace that best approximates the input points. More precisely, for a given $p\geq 1$, we aim to minimize the $p$th power of the $\ell_p$ norm of the error vector $(\|a_1-\bm{P}a_1\|,\ldots,\|a_n-\bm{P}a_n\|)$, where $\bm{P}$ denotes the projection matrix onto the subspace and the norms are Euclidean.
In \emph{constrained} subspace approximation (CSA), we additionally have constraints on the projection matrix $\bm{P}$. In its most general form, we require $\bm{P}$ to belong to a given subset $\mathcal{S}$ that is described explicitly or implicitly.

We introduce a general framework for constrained subspace approximation. Our approach, that we term coreset-guess-solve, yields either $(1+\varepsilon)$-multiplicative or $\varepsilon$-additive approximations for a variety of constraints. We show that it provides new algorithms for partition-constrained subspace approximation with applications to {\it fair} subspace approximation, $k$-means clustering, and projected non-negative matrix factorization, among others. Specifically, while we reconstruct the best known bounds for $k$-means clustering in Euclidean spaces, we improve the known results for the remainder of the problems. 
\end{abstract}


\section{Introduction}\label{sec:intro}

Large data sets, often represented as collections of high-dimensional points, naturally arise in fields such as machine learning, data mining, and computational geometry. Despite their high-dimensional nature, these points typically exhibit low intrinsic dimensionality. Identifying (or summarizing) this underlying low-dimensional structure is a fundamental algorithmic question with numerous applications to data analysis. We study a general formulation, that we call the  {\em subspace approximation problem}. 

In subspace approximation, given a set of $n$ points $\{a_1,\ldots, a_n\} \in \mathbb{R}^d$ and a rank parameter $k$, we consider the problem of ``best approximating'' the input points with a $k$-dimensional subspace in a high-dimensional space. Here the goal is to find a rank $k$ projection $\bm{P}$ that minimizes the projection costs $\| a_i - \bm{P} a_i\|$, aggregated over $i\in [n]$. The choice of aggregation leads to different well-studied formulations. In the $\ell_p$ subspace approximation problem, the objective is $\left( \sum_i \| a_i - \bm{P} a_i \|_2^p \right)$. Formally, denoting by $A$ the $d\times n$ matrix whose $i$th column is $a_i$, the $\ell_p$-subspace approximation problem asks to find a rank $k$ projection matrix $\bm{P} \in \RR^{d\times d}$ that minimizes $\| \bm{A} - \bm{P} \bm{A} \|_{2,p}^p := \sum_{i=1}^n \|a_i-\bm{P}a_i\|_2^p$. For different choices of $p$, $\ell_p$-subspace approximation captures some well-studied problems, notably the {\em median hyperplane problem} (when $p=1$), the {\em principal component analysis (PCA) problem} (when $p=2$), and the {\em center hyperplane problem} (when $p=\infty$).  

Subspace approximation for general $p$ turns out to be NP-hard for all $p \ne 2$. For $p>2$, semidefinite programming helps achieve a constant factor approximation (for fixed $p$) for the problem~\cite{deshpande2011alg}. Matching hardness results were also shown for the case $p>2$, first assuming the Unique Games Conjecture~\cite{deshpande2011alg}, and then based only on $\text{P} \ne \text{NP}$~\cite{guruswami2016bypass}. For $p<2$, hardness results were first shown in the work of \cite{clarkson2015input}. 

Due to the ubiquitous applications of subspace approximation in various domains, several ``constrained'' versions of the problem have been extensively studied as well~\cite{drineas2004clustering,yuan2013truncated,papailiopoulos2013sparse,asteris2014nonnegative,boutsidis2014randomized,cohen2015dimensionality}. In the most general setting of the {\em constrained $\ell_p$-subspace approximation} problem, we are additionally given a collection $\mathcal{S}$ of rank-$k$ projection matrices (specified either explicitly or implicitly) and the goal is to find a projection matrix $\bm{P} \in \mathcal{S}$ minimizing the objective. I.e.,
\begin{align}\label{eq:csa}
    \min_{\bm{P} \in \mathcal{S}} \|\bm{A} - \bm{P}\bm{A}\|_{2,p}^p.
\end{align}

Some examples of problems in constrained subspace approximation include the well-studied {\em column subset selection} \cite{boutsidis2009improved, tropp2009column, deshpande2010efficient, civril2012column,guruswami2012optimal,boutsidis2014near,altschuler2016greedy} where the projection matrices are constrained to project on to the span of $k$ of the original vectors, {\em $(k,z)$-means clustering} in which the set of projection matrices can be specified by the partitioning of the points into $k$ clusters (see \cite{cohen2015dimensionality} for a reference), and many more which we will describe in this paper.

\subsection{Our Contributions and Applications}
In this paper, we provide a general algorithmic framework for constrained $\ell_p$-subspace approximation that yields either $(1+\varepsilon)$-multiplicative or $\varepsilon$-additive error approximations to the objective (depending on the setting), with running time  exponential in $k$. We apply the framework to several classes of constrained subspace approximation, leading to new results or results matching the state-of-the-art for these problems. Note that since the problems we consider are typically APX-hard (including $k$-means, and even the \emph{unconstrained} version of $\ell_p$-subspace approximation for $p > 2$), a running time exponential in $k$ is necessary for our results, assuming the Exponential Time Hypothesis; a discussion in Section~\ref{sec:prelims}. Before presenting our results, we start with an informal description of the framework. 

\vspace{-6pt}
\subparagraph*{Overview of Approach.} Our approach is based on coresets~\cite{feldman2007ptas} (also~\cite{kmeans-coreset-21, cohenkmeanscoreset2022, Huang2024coresetoptimallb} and references therein), but turns out to be different from the standard approach in a subtle yet important way. Recall that a (strong) coreset for an optimization problem $\mathcal{O}$ on set of points $\bm{A}$ is a subset $\bm{B}$ such that for any solution for $\mathcal{O}$, the cost on $\bm{B}$ is approximately the same as the cost on $\bm{A}$, up to an appropriate scaling. In the formulation of $\ell_p$-subspace approximation above, a coreset for a dataset $\bm{A}$ would be a subset $\bm{B}$ of its columns with $k' \ll n$ columns, such that for all $k$-dimensional subspaces, each defined by some $\bm{P}$, $\| \bm{B} - \bm{P} \bm{B}\|_{2,p}^p \approx \|\bm{A} - \bm{P} \bm{A}\|_{2,p}^p$, up to scaling. Thus the goal becomes to minimize the former quantity.

In the standard coreset approach, first a coreset is obtained, and then a problem-specific enumeration procedure is used to find a near optimal solution $\bm{P}$. For example, for the $k$-means clustering objective, one can consider all the $k$-partitions of the points in the coreset $\bm{B}$; each partition leads to a set of candidate centers, and the best of these candidate solutions will be an approximate solution to the full instance. Similarly for (unconstrained) $\ell_p$-subspace approximation, one observes that for an optimal solution, the columns of $\bm{P}$ must lie in the span of the vectors of $\bm{B}$, and thus one can enumerate over the combinations of the vectors of $\bm{B}$. Each combination gives a candidate $\bm{P}$, and the best of these candidate solutions is an approximate solution to the full instance.

However, this approach does not work in general for constrained subspace approximation. To see this, consider the very simple constraint of having the columns of $\bm{P}$ coming from some given subspace $S$. Here, the coreset for $\ell_p$-subspace approximation on $\bm{A}$ will be some set $\bm{B}$ that is ``oblivious'' of the subspace $S$. Thus, enumerating over combinations of $\bm{B}$ may not yield any vectors in $S$! 

\emph{Our main idea} is to avoid enumeration over candidate solutions, but instead, we view the solution (the matrix $\bm{P} \in \RR^{d \times k}$) as simply a set of variables. We then note that since the goal is to use $\bm{P}$ to approximate $\bm{B}$, there must be some combination of the vectors of $\bm{P}$ (equivalently, a set of $k$ coefficients) that approximates each vector $a_i$ in $\bm{B}$. If the coreset size is $k'$, there are only $k \cdot k'$ coefficients in total, and we can thus hope to enumerate these coefficients in time $\exp(k \cdot k')$. For every given choice of coefficients, we can then solve an optimization problem to find the optimal $\bm{P}$. For the constraints we consider (including the simple example above), this problem turns out to be convex, and can thus be solved efficiently!

This simple idea yields $\varepsilon$-additive approximation guarantees for a range of problems. We then observe that in specific settings of interest, we can obtain $(1+\varepsilon)$-multiplicative approximations by avoiding guessing of the coefficients. In these settings, once the coefficients have been guessed, there is a \emph{closed form} for the optimal basis vectors, in the form of low degree polynomials of the coefficients. We can then use the literature on solving polynomial systems of equations (viewing the coefficients as variables) to obtain algorithms that are more efficient than guessing. The framework is described more formally in Section~\ref{sec:framework}. 

We believe our general technique of using coresets to reduce the number of {\it coefficients} needed in order to turn a constrained non-convex optimization problem into a convex one, may be of broader applicability. We note it is fundamentally different than the ``guess a sketch'' technique for variable reduction in \cite{RSW16,BBBKLW19,regW19,mw20} and the techniques for reducing variables in non-negative matrix factorization \cite{m16}. To support this statement, the guess a sketch technique requires the existence of a small sketch, and consequently has only been applied to approximation with entrywise $p$-norms for $p \leq 2$ and weighted variants \cite{RSW16,BBBKLW19,mw20}, whereas our technique applies to a much wider family of norms. 
\vspace{-6pt}
\paragraph{Relation to Prior Work.}
We briefly discuss the connection to prior work on binary matrix factorization using coresets.
The work of \cite{VVWZ-23-bin_mat_fac} addresses binary matrix factorization by constructing a strong coreset that reduces the number of distinct rows via importance sampling, leveraging the discrete structure of binary inputs. 
Our framework generalizes these ideas to continuous settings: we use strong coresets not merely to reduce distinct rows, but to reduce the number of variables in a polynomial system for solving continuous constrained optimization problems. 
This enables us to extend the approach to real-valued matrices and to more general loss functions.
Overall, our framework can be seen as a generalization and unification of prior coreset-based ``guessing'' strategies, adapting them to significantly broader settings.

\vspace{-6pt}
\subparagraph*{Applications.} We apply our framework to the following applications.  Each of these settings can be viewed as subspace approximation with a constraint on the subspace (i.e., on the projection matrix), or on properties of the associated basis vectors. Below we describe these applications, mention how they can be formulated as Constrained Subspace Approximation, and state our results for them. See Table \ref{table:results} for a summary.

\begin{table}[ht]
\centering
\resizebox{\textwidth}{!}{
{\renewcommand{\arraystretch}{1.0}%
\begin{tabular}{|l|c|c|c|}
\hline
\multicolumn{1}{|c|}{\textbf{Problem}} & \textbf{Running Time} & \textbf{Approx.} & \rule{0pt}{12pt} \shortstack{\textbf{Prior Work}}\\
\hline
\multirow{2}{*}{PC-$\ell_p$-Subspace Approx.} & \rule{0pt}{12pt} $(\frac{\kappa}{\varepsilon})^{\poly(\frac{k}{\varepsilon})}\cdot \poly(n)$ (\ref{thm:additive-subspace-approx}) & $\left(O(\varepsilon p)\cdot \|\bm{A}\|_{p,2}^p\right)^+$ & - \\
\cline{2-2}
& \rule{0pt}{12pt} $n^{O(\frac{k^2}{\varepsilon})} \cdot \poly(H)$ (\ref{thm:multiplicative})&  $(1+\varepsilon)^{*}$  & - \\
\hline
\multirow{2}{*}{Constrained Subspace Est.} &  \rule{0pt}{12pt} $\poly(n)\cdot (\frac{1}{\delta})^{O(\frac{k^2}{\varepsilon})}$ (\ref{cor:CSE-additive}) & $(1+\varepsilon, O(\delta \cdot \|\bm{A}\|_F^2))^{\dagger}$ &  $\sim$ \\
\cline{2-2}
& \rule{0pt}{12pt} $O(\frac{nd\gamma}{\varepsilon})^{O(\frac{k^3}{\varepsilon})}$ (\ref{thm:CSE-multiplicative})& $(1+\varepsilon)^{*}$  &   $\sim$\\
\hline
\multirow{2}{*}{PNMF} &\rule{0pt}{12pt}  $O(\frac{dk^2}{\varepsilon})\cdot (\frac{1}{\delta})^{O(\frac{k^2}{\varepsilon})}$(\ref{thm:NMF-additive}) & $(1+\varepsilon, O(\delta \cdot \|\bm{A}\|_F^2))^{\dagger}$ & $\sim$  \\
\cline{2-2}
&\rule{0pt}{12pt}  $(\frac{nd\gamma}{\varepsilon})^{O(\frac{k^3}{\varepsilon})}$ (\ref{thm:NMF-multiplicative}) &  $(1+\varepsilon)^{*}$  &  $\sim$ \\
\hline 
\multicolumn{1}{|l|}{$k$-Means Clustering} & $O(nnz(\bm{A})+ 2^{\widetilde{O}(\frac{k}{\varepsilon})}+n^{o(1)})$ (\ref{thm:k-means-runtime}) & $(1+\varepsilon)^{*}$  & \cite{feldman2007ptas} \\
\hline
\multicolumn{1}{|l|}{Sparse PCA}&  $d^{O(\frac{k^3}{\varepsilon^2})}\cdot \frac{k^3}{\varepsilon}$ (\ref{thm:sparse-PCA}) & $(\varepsilon\|\bm{A}-\bm{A}_k\|_F^2)^{+}$  & \cite{del2022sparse} \\
\hline
\end{tabular}
}}
\caption{Summary of the upper bound results we get using our framework. In the approximation column, we use super scripts $*,+,\dagger$ to represent if its a multiplicative, additive, or multiplicative-additive approximation respectively. In the notes on prior work column, we use tilde ($\sim$) to indicate that no prior theoretical guarantees are known (only heuristics) and  hyphen ($-$) to specify that the problem is new. 
}\label{table:results}
\end{table}
\subsubsection{Subspace Approximation with Partition Constraints} 

First, we study a generalization of $\ell_p$-subspace approximation, where we have {\em partition constraints} on the subspace. 
More specifically, we consider {\em\fairSAp{p}}, where besides the point set $\{a_1,\cdots, a_n\}\in\mathbb{R}^d$, we are given $\ell$ subspaces $S_1, \cdots, S_\ell$ along with capacities $k_1, \cdots, k_\ell$ such that $\sum_{i=1}^\ell k_i = k$. Now the set of valid projections $\mathcal{S}$ is implicitly defined to be the set of projections onto the subspaces that are obtained by selecting $k_i$ vectors from $S_i$ for each $i\in [\ell]$, taking their span.

\fairSAp{p} can be viewed as a variant of data summarization with ``fair representation''. Specifically, when $S_i$ is the span of the vectors (or points) in group $i$, then by setting $k_i$ values properly (depending on the application or the choice of policy makers),~\fairSAp{} captures the problem of finding a summary of the input data in which groups are fairly represented. This corresponds to the equitable representation criterion, a popular notion studied extensively in the fairness of algorithms, e.g., clustering~\cite{kleindessner2019fair,jones2020fair,chiplunkar2020solve,hotegni2023approximation}.\footnote{We note that the fair representation definitions differ from those in the line of work on fair PCA and column subset selection~\cite{samadi2018price,tantipongpipat2019multi,matakos2023fair,song2024socially}, where the \emph{objective contributions} (i.e., projection costs) of different groups must either be equal (if possible) or ensure that the maximum incurred cost is minimized. 
We focus on the question of groups having equal, or appropriately bounded, \emph{representation} among the chosen low-dimensional subspace (i.e., directions). This distinction is also found in algorithmic fairness studies of other problems, such as clustering.}
We show the following results for PC-subspace approximation: 
\begin{itemize}
    \item First, in~\Cref{thm:additive-subspace-approx}, we show for any $p\ge 1$, an algorithm for \fairSAp{p} with runtime $(\frac{\kappa}{\varepsilon})^{\poly({k}/{\varepsilon})}\cdot \poly(n)$  that returns a solution with additive error at most $O(\varepsilon p)\cdot\|\bm{A}\|^p_{p,2}$, where $\kappa$ is the condition number of the optimal choice of vectors from the given subspaces.

    \item For $p=2$, which is one of the most common loss functions for \fairSAp{p}, we also present a multiplicative approximation guarantee. There exists a $(1+\varepsilon)$-approximation algorithm running in time $s^{O(k^2/\varepsilon)} \cdot \poly(H)$ where $H$ is the bit complexity of each element in the input and $s$ is the sum of the dimensions of the input subspaces $S_1, \cdots, S_\ell$, i.e., $s = \sum_{j=1}^\ell \dim(S_j)$. The formal statement is in Theorem~\ref{thm:multiplicative}. 

\end{itemize}

\subsubsection{Constrained Subspace Estimation}

The {\em Constrained Subspace Estimation} problem originates from the signal processing community~\cite{santamaria2017constrained}, and aims to find a subspace $V$ of dimension $k$, that best approximates a collection of experimentally measured subspaces $T_1,\cdots,T_m$, 
with the constraint that it intersects a model-based subspace $W$ in at least a predetermined number of dimensions $\ell$, i.e., $\textnormal{dim}(V \cap W)\geq \ell$. This problem arises in applications such as beamforming, where the model-based subspace is used to encode the available prior information about the problem. The paper of \cite{santamaria2017constrained} formulates and motivates that problem, and further present an algorithm based on a semidefinite relaxation of this non-convex
problem, where its performance is only
demonstrated via numerical simulation.

We show in Section \ref{sec:subspace-estimation}, that this problem can be reduced to at most $k$ instances of \fairSAp{2}, in which the number of parts is $2$. This will give us the following result for the constrained subspace estimation problem.
\begin{itemize}
    \item In \Cref{cor:CSE-additive}, we show a $(1+\eps, \delta\|A\|_F^2)$-multiplicative-additive approximation in time $\poly(n)\cdot (1/\delta)^{O(k^2/\varepsilon)}$.  
    \item In \Cref{thm:CSE-multiplicative}, we show a $(1+\varepsilon)$ multiplicative approximation in time $O(nd\gamma/\varepsilon)^{O(k^3/\varepsilon)}$ where we assume $A$ has integer entries of absolute value at most $\gamma$. We assume that $\gamma=\poly(n)$. 
\end{itemize}

\subsubsection{Projective Non-Negative Matrix Factorization}

Projective Non-Negative Matrix Factorization (PNMF)~\cite{yuan2005projective} (see also~\cite{Yuan2009ProjectiveNM, Yang2010LinearAN}) is a variant of Non-Negative Matrix Factorization (NMF), used for dimensionality reduction and data analysis, particularly for datasets with non-negative values such as images and texts. In NMF, a non-negative matrix $\bm{X}$ is factorized into the product of two non-negative matrices $\bm{W}$ and $\bm{H}$ such that $\bm{X} \approx \bm{WH}$ where $\bm{W}$ contains basis vectors, and $\bm{H}$ represents coefficients. 
In PNMF, the aim is to approximate the data matrix by projecting it onto a subspace spanned by non-negative vectors, similar to NMF. However, in PNMF, the factorization is constrained to be {\em projective}.

Formally, PNMF can be formulated as a constrained $\ell_2$-subspace approximation as follows: the set of feasible projection matrices $\mathcal{S}$, consists of all matrices that can be written as $\bm{P} = UU^T$, where $U$ is a $d \times k$ orthonormal matrix with all non-negative entries.

We show the following results:
\begin{itemize}
    \item In \Cref{thm:NMF-additive}, we show a $(1+\eps, \delta\|A\|_F^2)$-multiplicative-additive approximation in time $O(dk^2/\varepsilon)\cdot (1/\delta)^{O(k^2/\varepsilon)}$.
    \item In \Cref{thm:NMF-multiplicative}, we show a $(1+\varepsilon)$ multiplicative approximation in time $(nd\gamma)^{O(k^3/\varepsilon)}$, where we assume $A$ has integer entries of absolute value at most $\gamma$. 
\end{itemize}

\subsubsection{$k$-Means Clustering}
$k$-means is a popular clustering algorithm widely used in data analysis and machine learning. Given a set of $n$ vectors $a_1,\cdots,a_n$ and a parameter $k$, the goal of $k$-means clustering is to partition these vectors into $k$ clusters $\{C_1,\cdots,C_k\}$ such that the sum of the squared distances of all points to their corresponding cluster center $\sum_{i=1}^n \|a_i-\mu_{C(a_i)}\|_2^2$ is minimized, where $C(a_i)$ denotes the cluster that $a_i$ belongs to and $\mu_{C(a_i)}$ denotes its center. It is an easy observation that once the clustering is determined, the cluster centers need to be the centroid of the points in each cluster. It is shown in \cite{cohen2015dimensionality} 
that this problem is an instance of constrained subspace approximation. More precisely, the set of valid projection matrices are all those that can be written as ${\bm P} = X_CX_C^T$, where $X_C$ is a $n\times k$ matrix where $X_C(i,j)$ is $1/\sqrt{|C_j|}$ if $C(a_i)=j$ and $0$ otherwise. Note that this is an orthonormal matrix and thus $X_C X_C^T$ is an orthogonal projection matrix. Further, note that using our language we need to apply the constrained subspace approximation on the matrix $A^T$, i.e., $\min_{\bm{P} \in \mathcal{S}}\|A^T-\bm{P}A^T\|_F^2$.

In \Cref{thm:k-means-runtime}, we show a $(1+\eps)$ approximation algorithm for $k$-means that runs in $O(nnz(\bm{A})+ 2^{\widetilde{O}(k/\varepsilon)}+n^{o(1)})$ time,  whose dependency on $k$ and $\varepsilon$ matches that of \cite{feldman2007ptas}.

\subsubsection{Sparse PCA}
The goal of Principal Component Analysis (PCA) is to find $k$ linear combinations of the $d$ features (dimensions), which are called principal components, that captures most of the mass of the data. As mentioned earlier, PCA is the subspace approximation problem with $p=2$. However, typically the obtained principal components are linear combinations of all vectors which makes interpretability of the components more difficult. As such, {\em Sparse PCA} which is the optimization problem obtained from PCA by adding a sparsity
constraint on the principal components have been defined which provides higher data interpretability \cite{del2022sparse, zou2006sparse, cadima1995loading, hastie2015statistical, boutsidis2011sparse}.

Sparse PCA can be formulated as a constrained subspace approximation problem in which the set of projection matrices are constrained to those that can be written as $P=UU^T$ where $U$ is a $d\times k$ orthonormal matrix such that the total number of non-zero entries in the $U$ is at most $s$, for a given parameter $s$.

We give an algorithm that runs in time $d^{O(k^3/\varepsilon^2)}\left(dk^3/\varepsilon+ d\log d\right)$ that computes a $\varepsilon \|\bm{A}-\bm{A}_k\|_F^2$ additive approximate solution, which translates to a $(1+\varepsilon)$-multiplicative approximate solution to one formulation the problem (see \Cref{thm:sparse-PCA} for the exact statement).

\subsubsection{Column Subset Selection with Partition Constraint}
Column subset selection (CSS) is a popular data summarization technique~\cite{boutsidis2014randomized,cohen2015dimensionality,altschuler2016greedy}, where given a matrix $\bm{A}$, the goal is to find $k$ columns in $\bm{A}$ that best approximates all columns of $\bm{A}$. Since in CSS, a subset of columns in the matrix are picked as the summary of the matrix $\bm{A}$, enforcing partition constraints naturally captures the problem of column subset selection with fair representation. 
More formally, in {\em column subset selection with partition constraints} ({\em\fairCSSx{}}), given a partitioning of the columns of $\bm{A}$ into $\ell$ groups, $\bm{A}^{(1)}, \cdots, \bm{A}^{(\ell)}$, along with capacities $k_1,\cdots,k_\ell$, where $\sum_i k_i=k$, the set of valid subspaces are obtained by picking $k_i$ vectors from $\bm{A}^{(i)}$, and projecting onto the span of these $k$ columns of $\bm{A}$.

In Section~\ref{sec:css-hardness}, 
we show that \fairCSSx{} is hard to approximate to any factor $f$ in polynomial time, even if there are only two groups, or even when we allow for violating the capacity constraint by a factor of $O(\log n)$
(see Theorem~\ref{thm:hard-CSS} for the formal statement). This is in sharp contrast with the standard column subset selection problem for which efficient algorithms with tight guarantees are known. 


\section{Preliminaries}\label{sec:prelims}
We will heavily use standard notations for vector and matrix quantities. For a matrix $\bm{M}$, we denote by $\bm{M}_{.,i}$ the $i$th column of $\bm{M}$ and by $\bm{M}_{i,.}$ the $i$th row. We denote by $\norm{\bm{M}}_F$ the Frobenius norm, which is simply $\sqrt{\sum_{i,j} m_{ij}^2}$, where $m_{ij}$ is the entry in the $i$th row and $j$th column of $\bm{M}$. We also use mixed norms, 
where $\norm{\bm{M}}_{2,p} = \left( \sum_i \norm{\bm{M}_{.,i}}_2^p \right)^{1/p}$. I.e., it is the $\ell_p$ norm of the vector whose entries are the $\ell_2$ norm of the columns of $\bm{M}$.

We also use $\sigma_{\min} (\bm{M})$ to denote the least singular value of a matrix, and $\sigma_{\max} (\bm{M})$ to denote the largest singular value. The value $\kappa (\bm{M})$ is used to denote the condition number, which is the ratio of the largest to the smallest singular value.

In analyzing the running times of our algorithms, we will use the following basic primitives, the running times of which we denote as $T_0$ and $T_1$ respectively. These are standard results from numerical linear algebra; while there are several improvements using randomization, these bounds will not be the dominant ones in our running time, so we do not optimize them.

\begin{lemma}[SVD Computation; see~\cite{golub2013matrix}]
\label{lem:svd-time}
Given $\bm{A}\in \RR^{d\times n}$, computing the reduced matrix $\bm{B}$ as in \Cref{lem:coreset_p=2} takes time $T_0 := H \cdot \min \{ O(nd^2), O(nd \cdot \frac{k}{\eps})\}$, where $H$ is the maximum bit complexity of any element of $\bm{A}$.
\end{lemma}

\begin{lemma}[Least Squares Regression; see~\cite{golub2013matrix}]
\label{lem:regression-time}
Given $\bm{A} \in \RR^{d\times n}$ and given a target matrix $\bm{B}$ with $r$ columns, the optimization problem $\min_{\bm{C}} \norm{\bm{B} - \bm{A} \bm{C}}_F^2$ can be solved in time $T_1 := O(nrd^2 \cdot $H$)$, where $H$ is the maximum bit length of any entry in $A, B$.
\end{lemma}

\subparagraph*{Remark on the Exponential in $k$ Running Times.} In all of our results, it is natural to ask if the exponential dependence on $k$ is necessary. We note that many of the problems we study are APX hard, and thus obtaining \emph{multiplicative} $(1+\eps)$ factors will necessarily require exponential time in the worst case. For problems that generalize $\ell_p$-subspace approximation (e.g., the \fairSAp{p} problem, Section~\ref{sec:subspace-approximation-main}), the works of~\cite{guruswami2016bypass} and \cite{clarkson2015input} showed APX hardness. In these reductions, we in fact have the stronger property that the YES and NO instances differ in objective value by $\frac{1}{\text{poly}(k)} \cdot \norm{\bm{A}}_{2,p}^p$, where $\bm{A}$ is the matrix used in the reduction. Thus, assuming the Exponential Time Hypothesis, even the additive error guarantee in general requires an exponential dependence on either $k$ or $1/\eps$.

\section{Framework for Constrained Subspace Approximation}\label{sec:framework}
Given a $d\times n$ matrix $\bm{A}$ and a special collection $\cS$ of rank $k$ projection matrices, we are interested in selecting the projection matrix $\bm{P}\in \cS$ that minimizes the sum of projection costs (raised to the $p^{\textnormal{th}}$ power) of the columns of $\bm{A}$ onto $\bm{P}$. More compactly, the optimization problem is 
\begin{align}
\label{prog:CSA}   \min\limits_{\bm{P} \in \cS}&:\, \|\bm{A}-\bm{P}\bm{A}\|_{2,p}^p. \tag{CSA}
\intertext{
A more geometric and equivalent interpretation is that we have a collection of $n$ data-points $\{a_1,a_2,\dots,a_n\}\subseteq \RR^d$ and we would like to approximate these data points by a subspace while satisfying certain constraints on the subspace:
}
\label{prog:CSA-geo}    \min&:\, \sumL_{i=1}^n \|a_i -\widehat{a}_i\|_2^p  \tag{CSA-geo} \\
    & \widehat{a}_i \in \textnormal{ColumnSpan}(\bm{P}) \nonumber\\
    &\bm{P}\in \cS. \nonumber
\end{align}
See \Cref{lem:csa-alt-equiv} for a proof of the equivalence. We provide a unified framework to obtain approximately optimal solutions for various special collections of $\cS$. 
In our framework, there are three steps to obtaining an approximate solution to any instance of \ref{prog:CSA}.
\begin{enumerate}
    \item \textbf{Build a coreset:} Reduce the size of the problem by replacing $\bm{A}$ with a different matrix $\bm{B}\in \RR^{d\times r}$ with fewer number of columns typically $\poly(k,1/\varepsilon)$. The property we need to guarantee is that the projection cost is approximately preserved possibly with an additive error $c\geq 0$ independent of $\bm{P}$:
    \begin{align}
    \label{eqn:approximate_cost}    \|\bm{B}-\bm{P}\bm{B}\|_{2,p}^p \in (1,1+\varepsilon)\cdot   \|\bm{A}-\bm{P}\bm{A}\|_{2,p}^p-c \quad \forall \bm{P} \textnormal{ with rank at most } k.
    \end{align}
    Such a $\bm{P}$ (for $p=2$) has been referred to as a \emph{Projection-Cost-Preserving Sketch with one sided error} in \cite{cohen2015dimensionality} . See \Cref{defn:strong_coreset}, \Cref{thm:coresets-lp}, and \Cref{lem:coreset_p=2} for results obtaining such a $\bm{B}$ for various $1\leq p<\infty$. \Cref{lem:approximate_soln} shows that approximate solutions to reduced instances $(\bm{B}, \cS)$ satisfying \Cref{eqn:approximate_cost} are also approximate solutions to the original instance $(\bm{A}, \cS)$.
    \item \textbf{Guess Coefficients:} Since the projection matrix $\bm{P}$ is of rank $k$, it can be represented as $\bm{U}\bm{U}^T$ such that $\bm{U}^T\bm{U}=\bm{I}_k$. Using this, observe that the residual matrix 
    \begin{align*}
    \bm{B}-\bm{P}\bm{B} = \bm{B}-\bm{U}(\bm{U}^T\bm{B})
    \end{align*}
     can be represented as $\bm{B}-\bm{U}\bm{C}$ where $\bm{C}=\bm{U}^T\bm{B}$ is a $\RR^{k\times r}$ matrix. The norm of the $i^{\textnormal{th}}$ column of $\bm{C}$ can be bounded by $\|b_i\|_2$ the norm of the $i^{\textnormal{th}}$ column of $\bm{B}$. This allows us to guess every column of $\bm{C}$ inside a $k$ dimensional ball of radius at most the norm of the corresponding column in $\bm{B}$. Using a net with appropriate granularity, we guess the optimal $\bm{C}$ up to an additive error. 
    \item \textbf{Solve:} For every fixed $\bm{C}$ in the search space above,  we solve the constrained regression problem 
    \begin{align*}
        \min\limits_{\bm{U}\in \RR^{d\times k}: \bm{U}\bm{U}^T\in \cS}\|\bm{B}-\bm{U}\bm{C}\|_{2,p}^p
    \end{align*}
    exactly. If $\widehat{\bm{C}}$ is the $\bm{C}$ matrix that induces the minimum cost, and $\widehat{\bm{U}}$ is the minimizer to the constrained regression problem, we return the projection matrix $\widehat{\bm{U}}\widehat{\bm{U}}^T$. 
\end{enumerate}

The following lemma formalizes the framework above and can be used as a black box application for several specific instances of \ref{prog:CSA}.
\begin{lemma}
\label{lem:structural-CSA}
Given an instance $(\bm{A}, \cS)$ of \ref{prog:CSA}, for $1\leq p<\infty$,
\begin{enumerate}
    \item Let $T_s$ be the time taken to obtain a smaller instance $(\bm{B}, \cS)$ such that the approximate cost property in \Cref{eqn:approximate_cost} is satisfied and the number of columns in $\bm{B}$ is $r$.
    \item Let $T_r$ be the time taken to solve the constrained regression problem for any fixed $\bm{B}\in \RR^{d\times r}$ and $\bm{C}\in \RR^{k\times r}$
    \begin{align}
  \label{eqn:constrained_regression}      \min\limits_{\bm{U}\in \RR^{d\times k}: \bm{U}\bm{U}^T\in \cS}\|\bm{U}\bm{C}-\bm{B}\|_{2,p}^p.
    \end{align}
Then for any granularity parameter $0<\delta<1$, we obtain a solution $\bm{P}\in \cS$ such that 
\begin{align}
   \|\bm{A}-\bm{P}\bm{A}\|_{2,p}^p \leq (1+\varepsilon)\opt +\Delta
\end{align}
in time $T_s+T_r\cdot O((1/\delta)^{kr})$.
\end{enumerate}
Here, $\Delta=(1+\varepsilon)\|\bm{A}\|_{2,p}^p\cdot\left((1+\delta)^p-1\right)$ and $\opt=\min\limits_{\bm{P}' \in \cS}\, \|\bm{A}-\bm{P}'\bm{A}\|_{2,p}^p$.
\end{lemma}
\begin{proof}
    Let the optimal solution to the instance $(\bm{A}, \cS)$ be $\bm{P}^* = \bm{U}^*{\bm{U}^*}^T$ and let $\bm{C}^*={\bm{U}^*}^T\bm{B}$. Since the columns of $\bm{U}^*$ are unit vectors, the norm of the $i^{\textnormal{th}}$ column of $\bm{C}^*$ is at most $\|b_i\|_2$ the norm of the $i^{\textnormal{th}}$ column of $\bm{B}$. We will try to approximately guess the columns of $\bm{C}^*$ using epsilon nets. For each $i$, we search for the $i^{th}$ column of $\bm{C}$ using a $(\|b_i\|_2\cdot \delta)$-net inside a $k$ dimensional ball of radius $\|b_i\|_2$ centered at origin. The size of the net for each column of $\bm{C}$ is $O((1/\delta)^k)$ and hence the total search space over matrices $\bm{C}$ has $O((1/\delta)^{kr})$ possibilities. 

    For each $\bm{C}$, we solve the constrained regression problem in \Cref{eqn:constrained_regression}. Let $\widehat{\bm{C}}$ be the matrix for which the cost is minimized and $\widehat{\bm{U}}$ be the corresponding minimizer to the constrained regression problem respectively. Consider the solution $\widehat{\bm{P}}=\widehat{\bm{U}}\widehat{\bm{U}}^T$. The cost of this solution on reduced instance $(\bm{B}, \cS)$ is 
    \begin{align}
        \|\bm{B}-\widehat{\bm{U}}\widehat{\bm{U}}^T\bm{B}\|_{2,p}^p &\leq  \|\bm{B}-\widehat{\bm{U}}\widehat{\bm{C}}\|_{2,p}^p. 
        \intertext{  Let $\overline{\bm{C}}$ be the matrix in the search space such that $\|\overline{\bm{C}}_{.,i} - \bm{C}^*_{.,i}\|_2 \leq \|b_i\|_2\cdot \delta$ for every $i\in [r]$. Using the cost minimality of $\widehat{\bm{C}}$, we can imply that the above cost is }
        &\leq \min\limits_{\bm{U}\in \RR^{d\times k}: \bm{U}\bm{U}^T\in \cS} \|\bm{B}-\bm{U}\overline{\bm{C}}\|_{2,p}^p\\
        & \leq \|\bm{B}-\bm{U}^*\overline{\bm{C}}\|_{2,p}^p. 
    \end{align}
It remains to upper bound the difference $\Delta=\|\bm{B}-\bm{U}^*\overline{\bm{C}}\|_{2,p}^p- \|\bm{B}-\bm{U}^*\bm{C}^*\|_{2,p}^p$. If we let $b_i^*:= (\bm{U}^*\bm{C}^*)_{.,i}$ and $\overline{b}_i:= (\bm{U}^*\overline{\bm{C}})_{.,i}$ for $i\in [r]$, then 
\begin{align}
    &\Delta = \sumL_{i=1}^r \left(\|b_i-\overline{b}_i\|_2^p- \|b_i-b_i^*\|_2^p\right).
    \intertext{Using the fact that $\|\overline{\bm{C}}_{.,i} - \bm{C}^*_{.,i}\|_2 \leq \|b_i\|_2\cdot \delta$, we know that} 
    \|\overline{b}_i- b_i^*\|_2 &= \|\bm{U}^*(\overline{\bm{C}}_{.,i} - \bm{C}^*_{.,i})\|_2\leq \|\overline{\bm{C}}_{.,i} - \bm{C}^*_{.,i}\|_2 \leq \|b_i\|_2\cdot \delta. 
\end{align}
     This implies that each error term 
\begin{align}
  \Delta_i &:= \|b_i-\overline{b}_i\|_2^p- \|b_i-b_i^*\|_2^p \\
  &\leq (\|b_i-b_i^*\|_2+ \|b_i^*-\overline{b}_i\|_2)^p- \|b_i-b_i^*\|_2^p \tag{Triangle inequality}\\
  &\leq (\|b_i-b_i^*\|_2+ \|b_i\|\cdot \delta)^p- \|b_i-b_i^*\|_2^p \tag{$\|\overline{b}_i- b_i^*\|_2\leq \|b_i\|_2\cdot \delta$}\\
  &\leq \|b_i\|_2^p\cdot \left((1+\delta)^p-1\right). \tag{$(x+\delta)^p-x^p$ is increasing in $[0,1], \,\|b_i-b_i^*\|_2\leq \|b_i\|_2$}
\end{align}
    Summing up, the total error $\Delta$ is at most $\|\bm{B}\|_{2,p}^p\cdot\left((1+\delta)^p-1\right)=O(\delta p)\cdot \|\bm{B}\|_{2,p}^p$ for $\delta\leq 1/p$. This implies that 
    \begin{align}
        \|\bm{B}-\widehat{\bm{P}}\bm{B}\|_{2,p}^p&\leq \|\bm{B}-\bm{P}^*\bm{B}\|_{2,p}^p+ \|\bm{B}\|_{2,p}^p\cdot\left((1+\delta)^p-1\right)
        \intertext{Using the property of $\bm{B}$ from \Cref{eqn:approximate_cost}, we can imply that}
        \|\bm{A}-\widehat{\bm{P}}\bm{A}\|_{2,p}^p& \leq (1+\varepsilon)\|\bm{A}-\bm{P}^*\bm{A}\|+ \|\bm{B}\|_{2,p}^p\cdot\left((1+\delta)^p-1\right).
\intertext{setting $\bm{P}=0$ in \Cref{eqn:approximate_cost} and using the fact that $c\geq 0$ gives $\|\bm{B}\|_{2,p}^p\leq (1+\varepsilon)\|\bm{A}\|_{2,p}^p$. Plugging this in the equation above gives}
\|\bm{A}-\widehat{\bm{P}}\bm{A}\|_{2,p}^p& \leq (1+\varepsilon)\|\bm{A}-\bm{P}^*\bm{A}\|+ (1+\varepsilon)\|\bm{A}\|_{2,p}^p\cdot\left((1+\delta)^p-1\right)
    \end{align}
    The total time taken by the algorithm is $T_s+T_r\cdot O((1/\delta)^{kr})$.
\end{proof}

\begin{lemma}
\label{lem:csa-alt-equiv}
The mathematical programs \ref{prog:CSA} and \ref{prog:CSA-geo} equivalent to the following ``constrained factorization'' problem:
\begin{align}
\label{prog:CSA-factorization}   \min\limits_{\bm{U}\bm{U}^T \in \cS,\, \bm{H}\in \RR^{d\times n}}&\, \|\bm{A}-\bm{U}\bm{H}\|_{2,p}^p. \tag{CSA-fac}
 \end{align}
\end{lemma}
\begin{proof}
First, we will prove the equivalence between \ref{prog:CSA} and \ref{prog:CSA-factorization}. 
        \begin{enumerate}
            \item The easier direction to see is $\min_{\bm{U}\bm{U}^T \in \cS,\, \bm{H}\in \RR^{d\times n}}\|\bm{A}-\bm{U}\bm{H}\|_{2,p}^p \leq \min_{\bm{U}\bm{U}^T\in \cS}\|\bm{A}-\bm{U}\bm{U}^T\bm{A}\|_{2,p}^p$ because setting $\bm{H}=\bm{U}^T\bm{A}$ in \ref{prog:CSA-factorization} gives \ref{prog:CSA}.
            \item For the other direction, it suffices to show that for any fixed choice of $\bm{U}$ such that $\bm{U}\bm{U}^T\in \cS$, an optimal choice of $\bm{H}$ is $\bm{U}^T\bm{A}$. In order to see this, observe that the problem 
            \begin{align}
                \min\limits_{\bm{H}}\|\bm{A}-\bm{U}\bm{H}\|_{2,p}^p &= \min\limits_{\bm{H}}\sumL_{i=1}^n\|a_i-\bm{U}h_i\|_2^p
                \intertext{where $a_i$ and $h_i$ are the $i^{th}$ columns of $\bm{A}$ and $\bm{H}$ respectively. Since the cost function decomposes into separate problems for each column, we can push the minimization inside. }
                &= \sumL_{i=1}^n\left(\min\limits_{h_i}\|a_i-\bm{U}h_i\|_2\right)^p.
            \end{align}
Using normal equation, the optimal choice for $h_i$ satisfies $\bm{U}^T\bm{U}h_i=\bm{U}^Ta_i$. Since the columns of $\bm{U}$ are orthonormal, this implies that $h_i=\bm{U}^Ta_i$ for each $i\in [n]$ and hence $\bm{H}=\bm{U}^T\bm{A}$.
        \end{enumerate}

Now we show the equivalence between \ref{prog:CSA-factorization} and \ref{prog:CSA-geo}. Observe that \ref{prog:CSA-geo} can be re-written as 
\begin{align*}
  \min&\, \sumL_{i=1}^n \|a_i -\widehat{a}_i\|_2^p  \\
    & \widehat{a}_i \in \textnormal{ColumnSpan}(\bm{U}) \nonumber\\
    &\bm{U}\bm{U}^T\in \cS. \nonumber
\end{align*}
Because the column span of $\bm{P}=\bm{U}\bm{U}^T$ is identical to the column span of $\bm{U}$. Replacing $\widehat{a}_i\in \textnormal{ColumnSpan}(\bm{U})$ by $\widehat{a}_i = \bm{U}h_i$ gives \ref{prog:CSA-factorization}.
\end{proof}

\begin{defn}[Strong coresets; as defined in \cite{WY24}]
\label{defn:strong_coreset}
    Let $1\leq p<\infty$ and $0<\varepsilon <1 $. Let $\bm{A}\in \RR^{d\times n}$. Then, a diagonal matrix $\bm{S}\in \RR^{n\times n}$ is a $(1\pm\varepsilon)$ strong coreset for $\ell_p$ subspace approximation if for all rank $k$ projection matrices $\bm{P}_F$, we have
    \begin{align}
\label{eqn:strong_coreset_property}        \|(\bm{I}-\bm{P}_F)\bm{A}\bm{S}\|_{2,p}^p \in  (1\pm \varepsilon)\|(\bm{I}-\bm{P}_F)\bm{A}\|_{2,p}^p.
    \end{align}
The number of non-zero entries  $\textnormal{nnz}(\bm{S})$ of $\bm{S}$ will be referred to as the size of the coreset. 
\end{defn}
\begin{theorem}[Theorems 1.3 and 1.4 of~\cite{WY25}] \label{thm:coresets-lp}
Let $p \in [1, 2) \cup (2, \infty)$ and $\eps >0$ be given, and let $\bm{A}\in \RR^{d\times n}$. There is an algorithm running in $\widetilde{O}(\textnormal{nnz}(\bm{A})+d^\omega)$ time which, with probability at least $1-\delta$, constructs a strong coreset $\bm{S}$ that satisfies \cref{defn:strong_coreset} and has size:
\begin{equation}\label{eq:coreset-size-lp}
\textnormal{nnz}(\bm{S}) = \begin{cases}
\frac{k}{\varepsilon^{4/p}}(\log(k/\varepsilon \delta))^{O(1)} \qquad \text{ if $p \in [1, 2)$}, \\
\frac{k^{p/2}}{\varepsilon^{p}}(\log(k/\varepsilon \delta))^{O(p^2)} \qquad \text{ if $p \in (2, \infty)$}.
\end{cases}
\end{equation}
\subparagraph*{Remark.} Note that for any $\bm{S}$ that satisfies the property in \Cref{defn:strong_coreset}, we can scale it up to satisfy  $\|(\bm{I}-\bm{P}_F)\bm{A}\bm{S}\|_{p,2}^p \in  (1, 1+\varepsilon)\|(\bm{I}-\bm{P}_F)\bm{A}\|_{p,2}^p$ matching the condition in \Cref{eqn:approximate_cost}.
\end{theorem}

For many of the applications, we have $p=2$. For this case, the choice of the reduced matrix $\bm{B}$ that replaces $\bm{A}$ is simply the matrix of scaled left singular vectors of $\bm{A}$. More formally, 
\begin{lemma}
\label{lem:coreset_p=2}
When $p=2$, if $\bm{A}=\sumL_{i=1}^n \sigma_i p_iq_i^T$ be the singular value decomposition of $\bm{A}$ (where $\sigma_i$ is the $i^{\text{th}}$ largest singular value and $p_i\in \RR^d,q_i\in \RR^n$ are the left singular vector and right singular vector corresponding to $\sigma_i$), then $\bm{B}=\sumL_{i=1}^{r} \sigma_i p_iq_i^T$ satisfies \Cref{eqn:approximate_cost} for $r=k+k/\varepsilon$.
\end{lemma}
\begin{proof}
For any two arbitrary projection matrices $\bm{P}$ and $\bm{P}'$ of rank $\leq k$, consider the difference 
\begin{align}
&   \left( \|\bm{A}-\bm{P}\bm{A}\|_F^2-\|\bm{B}-\bm{P}\bm{B}\|_F^2\right)-\left(\|\bm{A}-\bm{P}'\bm{A}\|_F^2-\|\bm{B}-\bm{P}'\bm{B}\|_F^2\right) \\
& =\langle\bm{A}\bm{A}^T, \bm{I}-\bm{P}\rangle-\langle\bm{B}\bm{B}^T, \bm{I}-\bm{P}\rangle-\langle\bm{A}\bm{A}^T, \bm{I}-\bm{P}'\rangle+\langle\bm{B}\bm{B}^T, \bm{I}-\bm{P}'\rangle \\
& = \langle\bm{A}\bm{A}^T-\bm{B}\bm{B}^T, \bm{P}'\rangle -\langle\bm{A}\bm{A}^T-\bm{B}\bm{B}^T, \bm{P}\rangle \\
& \leq \langle\bm{A}\bm{A}^T-\bm{B}\bm{B}^T, \bm{P}'\rangle \tag{$\bm{A}\bm{A}^T-\bm{B}\bm{B}^T\succeq 0,\, \bm{P}\succeq 0$}\\
&\leq \sumL_{i=r+1}^{r+k} \sigma_i \tag{rank of $\bm{P}' \leq k$ }\\
            & \leq k\cdot \sigma_r \tag{$\sigma_r\geq \sigma_{r'},\, r'\geq r$}\\
            &\leq \frac{k}{r-k}\cdot \left(\sumL_{i=k+1}^r \sigma_i\right) \tag{$\sigma_r \leq \sigma_{r'},\, r'\leq r$}\\
            &\leq \frac{k}{r-k}\|\bm{A}-\bm{A}_k\|_F^2 =\varepsilon \|\bm{A}-\bm{A}_k\|_F^2. \tag{$\|\bm{A}-\bm{A}_k\|_F^2 =\sumL_{i=k+1}^d \sigma_i$}
\end{align}
If we let $c:=\max_{\textnormal{rank}(\bm{P})\leq k}\left( \|\bm{A}-\bm{P}\bm{A}\|_F^2-\|\bm{B}-\bm{P}\bm{B}\|_F^2\right)$, then we have  
\begin{align*}
     c-\varepsilon\|\bm{A}-\bm{A}_k\|_F^2 \leq \|\bm{A}-\bm{P}\bm{A}\|_F^2-\|\bm{B}-\bm{P}\bm{B}\|_F^2 \leq c
\end{align*}
for any projection matrix $\bm{P}$ of rank at most $k$. This can we re written as 
\begin{align}
\label{eqn:approximate_cost_2}
\|\bm{B}-\bm{P}\bm{B}\|_F^2 \in (0,\varepsilon)\cdot\|\bm{A}-\bm{A}_k\|_F^2+ \|\bm{A}-\bm{P}\bm{A}\|_F^2 -c.
\end{align}
Using the fact that $\|\bm{A}-\bm{A}_k\|_F^2\leq \|\bm{A}-\bm{P}\bm{A}\|_F^2$, we get 
\begin{align*}
\|\bm{B}-\bm{P}\bm{B}\|_F^2 \in (1,1+\varepsilon)\cdot \|\bm{A}-\bm{P}\bm{A}\|_F^2 -c.
\end{align*}
The fact that $c\geq 0$ follows from the fact that 
\begin{align}
    \|\bm{A}-\bm{P}\bm{A}\|_F^2-\|\bm{B}-\bm{P}\bm{B}\|_F^2&= \langle\bm{A}\bm{A}^T- \bm{B}\bm{B}^T, \bm{I}-\bm{P} \rangle\\
    &\geq 0. \tag{$\bm{A}\bm{A}^T -\bm{B}\bm{B}^T \succeq 0, \; \bm{I}-\bm{P}\succeq 0$}
\end{align}
\end{proof}

\begin{remark}
\label{rem:stronger_PCPS_p=2}
    Notice that when $p=2$, \Cref{lem:coreset_p=2} proves the condition in \Cref{eqn:approximate_cost_2}:
    \begin{align*}
\|\bm{B}-\bm{P}\bm{B}\|_F^2 \in (0,\varepsilon)\cdot\|\bm{A}-\bm{A}_k\|_F^2+ \|\bm{A}-\bm{P}\bm{A}\|_F^2 -c        
    \end{align*}
    which is stronger than the condition in \Cref{eqn:approximate_cost}.
\end{remark} 
\begin{lemma}
    \label{lem:approximate_soln}
    If $(\bm{A}, \cS)$ is an instance of \ref{prog:CSA} and $\bm{B}\in \RR^{d\times r}$ is a matrix that satisfies \Cref{eqn:approximate_cost}, and  
    \begin{align}
        \widehat{\bm{P}}:= \argmin_{\bm{P}\in \cS}\|\bm{B}-\bm{P}\bm{B}\|_{2,p}
^p ,\quad \bm{P}^*:= \argmin_{\bm{P}\in \cS}\|\bm{A}-\bm{P}\bm{A}\|_{2,p}
^p,
\end{align}
then $\widehat{\bm{P}}$ is an $(1+\varepsilon)$-approximate solution to the instance $(\bm{A}, \cS)$ i.e.,
\begin{align}
   \|\bm{A}-\widehat{\bm{P}}\bm{A}\|_{2,p}
^p\leq (1+\varepsilon)\|\bm{A}-\bm{P}^*\bm{A}\|_{2,p}
^p.
\end{align}
\begin{enumerate}
    \item More generally, if $\widehat{\bm{P}}$ is an approximate solution to $(\bm{B}, \cS)$ such that 
\begin{align*}
    \|\bm{B}-\widehat{\bm{P}}\bm{B}\|_{2,p}^p \leq \alpha \|\bm{B}-\bm{P}\bm{B}\|_{2,p}^p +\beta \quad \forall \bm{P} \in \cS,
\end{align*}
 for some $\alpha\geq 1,\, \beta\geq 0$, then we have 
\begin{align*}
    \|\bm{A}-\widehat{\bm{P}}\bm{A}\|_{2,p}
^p \leq \alpha(1+\varepsilon)\|\bm{A}-\bm{P}^*\bm{A}\|_{2,p}^p + \beta.
\end{align*}
\item For the specific case when $p=2$, if $\widehat{\bm{P}}$ is an exact solution to $(\bm{B}, \cS)$, then we have 
\begin{align*}
     \|\bm{A}-\widehat{\bm{P}}\bm{A}\|_F^2
 \leq \|\bm{A}-\bm{P}^*\bm{A}\|_F^2 + \varepsilon\|\bm{A}-\bm{A}_k\|_F^2.
\end{align*}
\end{enumerate}

\end{lemma}
\begin{proof}
\begin{enumerate}
    \item Using the approximate optimality of $\widehat{\bm{P}}$ for the instance $(\bm{B}, \cS)$, we have 
    \begin{align}
      \|\bm{B}-\widehat{\bm{P}}\bm{B}\|_{2,p}^p &\leq \alpha\|\bm{B}-\bm{P}^*\bm{B}\|_{2,p}^p+\beta.
      \intertext{Using the lower-bound and upper-bound from \Cref{eqn:approximate_cost} for the LHS and RHS, we get}
      \|\bm{A}-\widehat{\bm{P}}\bm{A}\|_{2,p}^p -c & \leq \alpha(1+\varepsilon)\|\bm{A}-\bm{P}^*\bm{A}\|_{2,p}^p-\alpha c+\beta.
\intertext{Since $\alpha \geq 1$ and $c\geq 0$, we get}
 \|\bm{A}-\widehat{\bm{P}}\bm{A}\|_{2,p}
^p &\leq \alpha(1+\varepsilon)\|\bm{A}-\bm{P}^*\bm{A}\|_{2,p}^p + \beta.
    \end{align}
    \item Using the optimality of $\widehat{\bm{P}}$ for the instance $(\bm{B}, \cS)$ for with $p=2$, we have 
    \begin{align}
        \|\bm{B}-\widehat{\bm{P}}\bm{B}\|_F^2 \leq \|\bm{B}-\bm{P}^*\bm{B}\|_F^2.
    \end{align}
    Using \Cref{rem:stronger_PCPS_p=2}, we know that $\|\bm{B}-\bm{P}\bm{B}\|_F^2 \in (0,\varepsilon)\cdot\|\bm{A}-\bm{A}_k\|_F^2+ \|\bm{A}-\bm{P}\bm{A}\|_F^2 -c$ for any rank $k$ projection matrix $\bm{P}$ for some $c\geq 0$ independent of $\bm{P}$ (see \Cref{eqn:approximate_cost_2}). Using this, we get 
    \begin{align*}
        \|\bm{A}-\widehat{\bm{P}}\bm{A}\|_F^2-c \leq  \|\bm{B}-\widehat{\bm{P}}\bm{B}\|_F^2 \leq \|\bm{B}-\bm{P}^*\bm{B}\|_F^2\leq \|\bm{A}-\bm{P}^*\bm{A}\|_F^2+\varepsilon\|\bm{A}-\bm{A}_k\|_F^2-c.
    \end{align*}
    Canceling out the $-c$ gives the inequality we claimed. 
\end{enumerate}
\end{proof}

\begin{lemma}[Lemma 4.1 in \cite{Numercal_LinearAlgebra_Woodruff}]
\label{lemma:subspace_error_lb}
If $n\times d$ matrix $\bm{A}$ has integer entries bounded in magnitude by $\gamma$, and has rank $\rho\geq k$, then the $k^{\textnormal{th}}$ singular value $\sigma_k$ of $\bm{A}$ has $|\log \sigma_k|=O(\log(nd\gamma))$ as $nd\rightarrow \infty$. This implies that $\|\bm{A}\|_F/\Delta_k\leq (nd\gamma)^{O(k/(\rho-k))}$ as $nd\rightarrow \infty$. Here $\Delta_k:= \|\bm{A}-\bm{A}_k\|_F$
\end{lemma}

\section{Applications}

We present several applications to illustrate our framework.
\subsection{Constrained Subspace Estimation \cite{santamaria2017constrained}}\label{sec:subspace-estimation}

In constrained subspace estimation, we are given a collection  of target subspaces $T_1,T_2,\dots,T_m$ and a model subspace $W$. The goal is to find a subspace $V$ of dimension $k$ such that $\textnormal{dim}(V \cap W)\geq \ell$ that maximizes the average overlap between the subspace $V$ and $T_1,\dots,T_m$. More formally, the problem can be formulated as mathematical program:
\begin{align}
\label{prog:subspace_estimation_max}
    \max&: \langle \overline{\bm{P}}_T, \bm{P}_V\rangle \tag{CSE-max}\\
    &\textnormal{dim}(V)=k,\;\textnormal{dim}(V \cap W)\geq \ell, \\
    & \overline{\bm{P}}_T =\frac{1}{m}\sumL_{i=1}^m\bm{P}_{T
    _i},\\
    &\bm{P}_{T_i} \textnormal{ and } \bm{P}_V \textnormal{ are the projection matrices onto the subspaces } T_i \textnormal{ and } V \textnormal{ respectively.}
\end{align}
Let us assume that the constraint $\textnormal{dim}(V \cap W)\geq \ell$ is actually an exact constraint $\textnormal{dim}(V \cap W)= \ell$ because we can solve for $k-\ell+1$ different cases $\textnormal{dim}(V \cap W)= i$ for each $\ell \leq i\leq k$. Since $\overline{\bm{P}}_T$ is a PSD matrix, let it be $\bm{A}\bm{A}^T$ for some $\bm{A}\in \RR^{d\times d}$. Changing the optimization problem from a maximization problem to a minimization problem, we get
\begin{align}
\label{prog:subspace_estimation_min}
    \min&: \langle \bm{A}\bm{A}^T, \bm{I}-\bm{P}_V\rangle = \|\bm{A}-\bm{P}_{V}\bm{A}\|_F^2 \tag{CSE-min}\\
    &\bm{P}_V \textnormal{ is the projection matrix onto } V\\
     &\textnormal{dim}(V)=k,\; \textnormal{dim}(V \cap W)= \ell.
\end{align}

\begin{lemma}
   The \ref{prog:subspace_estimation_min} problem is a special case of \ref{prog:CSA}.  
\end{lemma}
\begin{proof}
 Setting $p=2$ and $\cS$ as the set of $k$ dimensional projection matrices $P_V$ such that $\textnormal{dim}(V \cap W)=\ell$ in \ref{prog:CSA} gives \ref{prog:subspace_estimation_min}.
\end{proof}

Let $\bm{B}\in \RR^{d\times r},\; r=k+k/\varepsilon$ be the reduced matrix obtained as in \Cref{lem:coreset_p=2}. Using \Cref{lem:approximate_soln}, it is sufficient to focus on the reduced instance with $\bm{A}$ replaced instead of $\bm{B}$. 

Any subspace $V$ such that $\textnormal{dim}(V)=k,\; \textnormal{dim}(V \cap W)= \ell$ can be represented equivalently as 
\begin{align*}
    V&= \textnormal{Span}(u_1,u_2,\dots,u_\ell, v_1,v_2,\dots,v_{k-\ell})\\
    &u_i \in W, \;v_j \in W^{\perp} \quad \forall i\in [\ell],\; j\in [k-\ell]. 
\end{align*}
Using these observations and \Cref{lem:csa-alt-equiv}, we can focus on the following subspace estimation program
\begin{align}
    \min&: \|\bm{B}-\bm{U}\bm{C}\|_F^2\\
\label{eqn:u_ortho-CSE}    &\bm{U} \textnormal{ is a orthogonal basis for } \textnormal{Span}(u_1,\dots,u_{\ell}, v_1,\dots,v_{k-\ell})\\
     &u_i \in W, \;v_j \in W^{\perp} \quad \forall i\in [\ell],\; j\in [k-\ell].
\intertext{
Since $\bm{C}$ is unconstrained, we can replace the condition in \Cref{eqn:u_ortho-CSE} with the much simpler condition $\bm{U}=[u_1,\dots,u_{\ell}, v_1,\dots,v_{\ell}]$. This gives 
}
\label{prog:subspace_estimation_min_simple}
    \min&: \|\bm{B}-\bm{U}\bm{C}\|_F^2 \tag{CSE-min-reduced}\\
   &\bm{U} =[u_1,\dots,u_{\ell}, v_1,\dots,v_{\ell}] \\
     &u_i \in W, \;v_j \in W^{\perp} \quad \forall i\in [\ell],\; j\in [k-\ell].
\end{align}

\begin{lemma}
    \label{lem:CSE-regression}
    For any fixed $\bm{B}\in \RR^{d\times r}$ and $\bm{C}\in \RR^{k \times r}$, the \Cref{prog:subspace_estimation_min_simple} can be solved exactly in $\poly(n)$ time.
\end{lemma}
\begin{proof}
For fixed $\bm{B}$ and $\bm{C}$, the objective is convex quadratic in $\bm{U}$ and the constraints are linear on $\bm{U}$. Linear constrained convex quadratic program can be efficiently solved. 
\end{proof}

\begin{cor}[Additive approximation for CSE]\label{cor:CSE-additive}
    Using \Cref{lem:structural-CSA}, we can get a subspace $V$ such that $\textnormal{dim}(V)=k,\; \textnormal{dim}(V \cap W)= \ell$ and 
    \begin{align*}
       \|\bm{A}-\bm{P}_V\bm{A}\|_F^2 \leq (1+\varepsilon)\opt + O(\delta \|\bm{A}\|_F^2)
   \end{align*}
    for any choice of $0<\delta<1$ in time $\poly(n)\cdot (1/\delta)^{O(k^2/\varepsilon)}$.  
\end{cor}


\Cref{lemma:subspace_error_lb} gives a lower bound for $\opt$ when the entries of the input matrix $\bm{A}$ are integers bounded in magnitude by $\gamma$.

\begin{theorem}[Multiplicative approximation for CSE]\label{thm:CSE-multiplicative}
    Given an instance $(\bm{A}\in \RR^{d\times n}, k, W)$ of constrained subspace estimation with integer entries of absolute value at most $\gamma$ in $\bm{A}$, there is an algorithm that obtains a subspace $V$ such that $\textnormal{dim}(V)=k,\; \textnormal{dim}(V \cap W)= \ell$ and 
    \begin{align*}
        \|\bm{A}-\bm{P}_V\bm{A}\|_F^2 \leq (1+\varepsilon)\opt
    \end{align*}
 in $ O(nd\gamma/\varepsilon)^{O(k^3/\varepsilon)}$ time. 
\end{theorem}
\begin{proof}
Using \Cref{lemma:subspace_error_lb}, we know that $\|\bm{A}\|_F^2/\|\bm{A}-\bm{A}_k\|_F^2\leq (nd\gamma)^{O(k)}$. Setting $\delta= \varepsilon\|\bm{A}-\bm{A}_k\|_F^2/\|\bm{A}\|_F^2\geq \varepsilon (nd\gamma)^{-O(k)}$ in \Cref{cor:CSE-additive} gives the desired time complexity.  
\end{proof}

\subsection{Partition Constrained $\ell_p$-Subspace Approximation}\label{sec:subspace-approximation-main}
We now consider the \fairSAp{p} problem, which generalizes the subspace approximation and subspace estimation problems. 

\begin{defn}[Partition Constrained $\ell_p$-Subspace Approximation]\label{def:constrained-subspace-selection}

In the \fairSAp{p} problem, we are given a set of target vectors $\{a_1,a_2,\dots,a_n\}\subseteq \RR^d$ as columns of a matrix $\bm{A}\in \RR^{d\times n}$, a set of $\ell$ subspaces $S_1,\dots,S_\ell \subseteq \RR^d$, and a sequence of capacity constraints $k_1, \cdots, k_\ell$ where $k_1 + \cdots + k_\ell = k$. The goal is to select $k$ vectors in total, $k_i$ from subspace $S_i$, such that their span captures as much of $\bm{A}$ as possible. Formally, the goal is to select vectors $\{v_{i,t_i}\}_{i\le \ell, t_i\le k_i}$, such that for every $i\le \ell$, $v_{i,1},\dots, v_{i,k_i} \in S_i$, so as to minimize $ \sum_{i\in [n]}\|\proj^{\bot}_{\textnormal{span}(\{v_{i,t_i}\}_{i\le \ell, t_i\le k_i})}(a_i)\|_2^p $.  
\end{defn}

Our results will give algorithms with running times exponential in $\poly(k)$ for \fairSAp{}. Given this goal, we can focus on the setting where $k_i=1$, since we can replace each $S_i$ in the original formulation with $k_i$ copies of $S_i$, with a budget of $1$ for each copy.

\subparagraph*{\fairSAp{} with Unit Capacity.}
Given a set of vectors $\{a_1,a_2,\dots,a_n\}\subseteq \RR^d$ as columns of a matrix $\bm{A}\in \RR^{d\times n}$ and subspaces $S_1,\dots,S_k \subseteq \RR^d$, select a vector $v_i \in S_i$ for $i\in [k]$ in order to minimize $\sum_{i\in [n]}\|\proj^{\perp}_{\textnormal{span}(v_1,\dots,v_k)}(a_i)\|_2^p$, where $p\ge 1$ is a given parameter. A more compact formulation is 
\begin{align}
  \label{prog:fairSAp-geo}  \min&: \sumL_{i=1}^n \|a_i-\widehat{a}_i\|_2^p  \tag{PC-$\ell_p$-SA-geo}\\
  &\widehat{a}_i \in \textnormal{Span}(v_1,\dots,v_k) \quad \forall i \in [n]\\
  & v_{j} \in S_{j} \quad \forall j \in [k].
\intertext{
Using \Cref{lem:csa-alt-equiv}, the two other equivalent formulations are  }
    \label{prog:fairSAp} \min&: \|\bm{A}-\bm{U}\bm{U}^T\bm{A}\|_{2,p}^p \tag{PC-$\ell_p$-SA}\\
\label{eqn:u_ortho-PCSA}    &\bm{U} \textnormal{ is an orthogonal basis for Span}(v_1,v_2,\dots,v_k) \\ 
    &v_{i}\in S_{i} \quad \forall i \in [k].\\
\label{prog:fairSAp-fac}    \min&: \|\bm{A}-\bm{V}\bm{C}\|_{2,p}^p \tag{PC-$\ell_p$-SA-fac}\\
   &\bm{V} =[v_1,\dots,v_{k}] \\
     &v_{i}\in S_i \quad \forall i \in [k].
\end{align}

In what follows, we thus focus on the unit capacity version. We can use our general framework to derive an additive error approximation, for any $p$. 
\newcommand{\bfB}{\bm{B}}

\subsubsection{Additive Error Approximation}\label{app:subspace-approximation-additive}

\begin{theorem}\label{thm:additive-subspace-approx}
There exists an algorithm for \fairSAp{p} with runtime $(\kappa/\varepsilon)^{\poly(k/\varepsilon)}\cdot \poly(n)$ which returns a solution with additive error at most $O(\varepsilon p) \cdot \|A\|_{p,2}^p$, where $\kappa$ is the condition number of an optimal solution $\bm{V}^*=\left [v_1^*,v_2^*,\dots,v_k^*\right]$ for the \fairSA{}problem \ref{prog:fairSAp-fac}.  
\end{theorem}

The algorithm we present below assumes a given bound on $\kappa$, the condition number. In practice, we can search for it via doubling and stop when a sufficiently small approximation error is reached or a certain time complexity is reached. We also note that it may happen that the \emph{optimal} solution uses a large $\kappa$, but there is an approximately optimal solution with small $\kappa$. In this case, our result can be applied with the smaller $\kappa$, and it gives a guarantee relative to the latter solution.

\begin{proof}
As a first step, we will find an additive approximation to the smaller instance obtained by replacing the $\bm{A}$ matrix with the smaller $\bm{B}$ matrix as in \Cref{lem:coreset_p=2}. Our proof mimics the argument from \Cref{lem:structural-CSA}, but we need a slight change in the analysis because $\{v_j\}$ are not orthogonal. Note that we can assume without loss of generality that the columns of $\bm{V}^*$ are unit vectors, $\sigma_{\max} (\bm{V}^*) \ge 1$ and $\sigma_{\min} (\bm{V}^*) \le 1$, and thus $\kappa \ge 1$. Given a bound on $\kappa$, the algorithm is simply the following: we first create a $\delta$-net for the Ball of radius $\kappa$ in $\RR^{k}$, with $\delta = \eps / \kappa$, and for each $i$, we form a guess for the coefficient vector $\bm{C}_{.,i}$ as $\norm{b_i}_2 \cdot u$, where $u$ is a vector from the net and $b_i$ is the $i^{\textnormal{th}}$ column of $\bm{B}$. For each guess $\widehat{\bm{C}}$, we solve for $\bm{V}$ that minimizes $\norm{\bm{B} - \bm{V} \widehat{\bm{C}}}_{2,p}$ subject to $v_j \in S_j$. Note that we can drop the unit vector constraints at this point; this makes the above optimization problem convex (specifically, it is the well-studied problem of $\ell_p$ regression~\cite{adil2024}), which can be solved in polynomial time.

To bound the error, we first note that the optimum coefficients $\bm{C}^*$ satisfy the condition that for each $i$, 
\[ \norm{\bm{C}^*_{., i}} \le \frac{\norm{b_i}_2}{\sigma_{\min}(\bm{V}^*)} \le \norm{b_i}_2 \cdot \kappa . \]
Now suppose we focus on one target vector $b_i$. By choice, in one of our guessed solutions, say $\overline{\bm{C}}$, we will have $\norm{\overline{\bm{C}}_{., i} - \bm{C}^*_{., i}} \le \norm{b_i}_2 \cdot \delta$. Thus, we have
\[  \norm{\overline{b}_i - b_i^*}_2 \le \norm{\bm{V}^* (\overline{\bm{C}}_{., i} - \bm{C}^*_{., i}) }_2 \le \sigma_{\max} (\bm{V}^*) \norm{\overline{\bm{C}}_{., i} - \bm{C}^*_{., i}}_2 \le \kappa \cdot \norm{b_i}_2 \cdot \delta. \]
Thus, analyzing the error $\Delta_i$ as in the proof of \Cref{lem:structural-CSA}, we obtain
\[ \Delta_i \le \norm{b_i}_2^p \cdot \left( (1+\varepsilon/p)^p -1 \right).  \]
This yields the desired additive guarantee to the reduced instance. Using the coreset property from \Cref{eqn:approximate_cost}, we know that the cost of the solution we find is at most $(1+\varepsilon)\opt+O(\varepsilon p)\cdot \|\bm{A}\|_{2,p}^p$. Using the fact that $\opt \leq \|\bm{A}\|_{2,p}^p$ completes the proof. 
\end{proof}

\subsubsection{Multiplicative Approximation Using Polynomial System Solving}\label{sec:fairsa-mult}
For the special case of $p=2$, it turns out that we can obtain a $(1+\eps)$-multiplicative approximation, using a novel idea.  

As described in our framework, we start by constructing  the reduced instance $\bfB, \cS$, where $\bfB = \{b_1, b_2, \dots, b_r\} \subset \RR^d$ is a set of target vectors and $\cS = \{ S_1, S_2,  \dots, S_k\}$ is the given collection of subspaces of $\RR^d$. We define $\bm{P}_j$ to be some fixed orthonormal basis for the space $S_j$. Recall that any solution to \fairSAp{2} is defined by (a) the vector $x_j$ that expresses the chosen $v_j$ as $v_j = \bm{P}_j x_j$ (we have one $x_j$ for each $j \in [k]$), and (b) a set of combination coefficients $c_{ij}$ used to represent the vectors $b_i$ using the vectors $\{v_j\}_{j=1}^k$. We collect the vectors $x_j$ into one long vector $\bm{x}$ and the coefficients $c_{ij}$ into a matrix $\bm{C}$. 

\begin{theorem}\label{thm:multiplicative}
Let $\bfB, \cS$ be an instance of \fairSAp{2}, where $\bfB = \{b_1, b_2, \dots b_r\}$, and suppose that the bit complexity of each element in the input is bounded by $H$. Suppose there exists an (approximately) optimal solution is defined by the pair $(\bm{x}^*, \bm{C}^*)$ with bit complexity $\text{poly}(n, H)$.  There exists an algorithm that runs in time $n^{O(k^2/\varepsilon)}\cdot poly(H)$ and outputs a solution whose objective value is within a $(1+\eps)$ factor of the optimum objective value. We denote $s=\sum_{j=1}^k s_j$ and $s_j =dim(S_j)$; $n$ for this result can be set to $\max(s,d,  k/\eps)$. 
\end{theorem}

\subparagraph*{Algorithm Overview.} Recall that $\bm{P}_j$ specifies an orthonormal basis for $S_j$. Let $\bm{P}_{ij}:=c_{ij}\bm{P}_j$, where $c_{ij}$ are variables.  Define $\bm{P}$ to be the $\RR^{rd\times s}$ matrix consisting of $r \times k$ blocks; the $(i,j)^{\textnormal{th}}$ block is $\bm{P}_{ij}$ and we let $\bm{x}, \bm{b}$ be the vectors representing all the $x_j, b_i$ stacked vertically respectively as shown below: 
\begin{align*}
   \bm{P}= \left[\begin{array}{c|c|c|c}
\bm{P}_{1,1}  & \bm{P}_{1,2}  & \cdots & \bm{P}_{1,k}   \\ \hline
\bm{P}_{2,1}  & \bm{P}_{2,2}  & \cdots & \bm{P}_{2,k}   \\ \hline
  \vdots & \vdots & \ddots & \vdots \\ \hline
 \bm{P}_{r,1}  & \bm{P}_{r,2} & \cdots & \bm{P}_{r,k}\\ 
\end{array}\right] , \quad 
\bm{x}=  \left[\begin{array}{c}
x_1   \\ \hline
x_2 \\ \hline
  \vdots  \\ \hline
x_k 
\end{array}\right], \quad 
\bm{b}= \left[\begin{array}{c}
b_1  \\ \hline
b_2 \\ \hline
  \vdots  \\ \hline
b_r  
\end{array}\right].
\end{align*} 
The problem \fairSAp{2} can now be expressed as the regression problem: 
\begin{align}
 \label{eqn:CSS-compact}   \min_{\bm{C,x}} : \|\bm{P}\bm{x}-\bm{b}\|_2^2.
\end{align}
Written this way, it is clear that for any $\bm{C}$, the optimization problem with respect to $\bm{x}$ is simply a regression problem.  For the sake of exposition, suppose that for the optimal solution $(\bm{C}^*, \bm{x}^*)$, the matrix $\bm{P}$ turns out to have a full column rank (i.e., $\bm{P}^T \bm{P}$ is invertible). In this case, the we can write down the normal equation $\bm{P}^T\bm{P}\bm{x}=\bm{P}^T\bm{b}$ and solve it using Cramer's rule! More specifically, let $\bm{D}=\bm{P}^T\bm{P}$ and $\bm{D}_j^{(i)}$ be the matrix obtained by replacing the $i^{\textnormal{th}}$ column in the $j^{\textnormal{th}}$ column block of $\bm{D}$ with the column $\bm{P}^T\bm{b}$ for $j\in [k], i\in [s_j]$. Using Cramer's rule, we have $x_j^{(i)}= \det(\bm{D}_j^{(i)})/\det(\bm{D})$. 

The key observation now is that substituting this back into the objective yields an optimization problem over (the variables) $\bm{C}$.  First, observe that using the normal equation, the objective can be simplified as
\[  \|\bm{P}\bm{x}-\bm{b}\|_2^2 = \bm{x}^T \bm{P}^T \bm{P} \bm{x} - \bm{x}^T \bm{P}^T \bm{b} - \bm{b}^T \bm{P} \bm{x} + \norm{\bm{b}}^2 = \norm{\bm{b}}^2 - \bm{b}^T \bm{P} \bm{x}. \]

Suppose $t$ is a real valued parameter that is a guess for the objective value. We then consider the following feasibility problem:
\begin{align}
    &\|\bm{P}\bm{x}-\bm{b}\|_2^2 =   \|\bm{b}\|_2^2  -\bm{b}^T\bm{P}\bm{x} \leq t \\
\label{eqn:polynomial-binsearch}    & \iff \|\bm{b}\|_2^2 -t  \leq  \sumL_{j\in [k], i\in [s_j]}(\bm{b}^T\bm{P})_j^{(i)}\det(\bm{D}_j^{(i)}) , \quad\det(\bm{D})=1. 
\end{align}
The idea is to solve this feasibility problem using the literature on solving polynomial systems. This leaves two main gaps: guessing $t$, and handling the case of $\bm{P}$ not having a full column rank in the optimal solution. We handle the former issue using known quantitative bounds on the solution value to polynomial systems, and the latter using a pre-multiplication with random matrices of different sizes.

\paragraph{Core Solver.}  We begin by describing the details of solving~\eqref{eqn:polynomial-binsearch}, assuming a feasible guess for $t$, and assuming that $\bm{P}$ has full column rank. Note that this is an optimization problem in the variables $c_{ij}$, and so the number of variables is $rk = O(k^2 / \eps)$. Furthermore, the degree of the polynomials is $O(s)$, and the bit sizes of all the coefficients is $\text{poly}(n, H)$.

We can thus use a well-known result on solving polynomial systems over the reals:
\begin{theorem}[\cite{renegar1992computationala,renegar1992computationalb,BPR96}]
\label{thm:poly_system_solving}
Given a real polynomial system $P(x_1,\dots,x_v)$ having $v$ variables and $m$ polynomial constraints $f_i(x_1,\dots,x_v) \Delta_i 0$, where $\Delta_i \in \{\geq , =, \leq \}$, where $d$ is the maximum degree of all polynomials, and $H$ is the maximum bit-size of the coefficients of the polynomials, one can find a solution to $P$ (or declare infeasibility) in time $(md)^{O(v)}\textnormal{poly}(H)$.
\end{theorem}

Applying this Theorem to our setting, we obtain a running time of $s^{O(rk)} \cdot \text{poly}(H) = n^{O(k^2/\varepsilon)}\cdot \textnormal{poly}(H)$, where $H$ is the maximum bit-size of the entries in the matrices $\bm{P}_j$ and the target vectors $\bm{b}$. 

\paragraph{Guessing $t$.} The solver step assumes that we are able to guess a feasible value of $t$. In order to perform a binary search, we need some guarantees on the range of the objective value. First, we note that the value of the objective is in the range $[0,\|\bm{b}\|_2^2]$, so we have an obvious upper bound. We can also test if the problem is feasible for $t=0$. If $t=0$ is infeasible, we need a non-trivial lower bound on $t$. Fortunately, this problem has been well-studied in the literature on polynomial systems. We will use the following Theorem:

\begin{theorem}[\cite{JPT13}]
\label{thm:poly_system_lowerbound}
Let $T=\{x\in \RR^v | f_1(x)\geq 0, \dots, f_{\ell}(x)\geq 0, f_{\ell+1}(x)=0,\dots, f_{m}(x)=0\}$ be the feasibility set for a polynomial system, where $f_1,\dots,f_m \in \ZZ[x_1,\dots,x_v]$ are polynomials with degrees bounded by an even integer $d$ and coefficients of absolute value at most $G$, and let $C$ be a compact connected component of $T$. Let $g\in \ZZ[x_1,\dots,x_v]$ be a polynomial of degree at most $d$ and coefficients of absolute value bounded by $H$. Then the minimum value that $g$ takes over $C$ if not zero, has absolute value at least 
    \begin{align*}
        (2^{4-v/2}\tilde{G}d^v)^{-v2^vd^v},
    \end{align*}
    where $\tilde{G}=\max\{G,2v+2m\}$.
\end{theorem}

We can use Theorem~\ref{thm:poly_system_lowerbound} to obtain a lower bound on the minimum non-zero value attainable for the polynomial
\begin{align}
\label{eqn:min_granularity}    \min_{\bm{C}}: \|\bm{b}\|_2^2 - \sumL_{j\in [k], i\in [s_j]}(\bm{b}^T\bm{P})_j^{(i)}\det(\bm{D}_j^{(i)}), \quad \text{over the set defined by } \det(\bm{D})=1.
\end{align}

Using our assumption that the bit complexity of the (near-) optimal solution $\bm{C}^*$ is $\Delta = \text{poly}(n, H)$, we can add box constraints for each of the variables, making the feasible set compact. Thus, we have polynomials of degree $O(s)$ defining the constraints and the objective, and the number of variables is $rk = O(k^2/\eps)$. Using Theorem~\ref{thm:poly_system_lowerbound}, we have that if the minimum value is non-zero, it must be at least
\[  \left( 2^{\text{poly}(n,H)} s^v   \right)^{-v 2^v s^v},  \text{ where } v = \frac{k^2}{\eps}.  \]
This implies that a binary search takes time $O\left(  (2s)^{k^2/\eps} \text{poly}(n,H) \frac{k^4}{\eps^2} \log s \right)$.  Note that this time roughly matches the complexity of the core Solver procedure.

\paragraph{Column Rank of $\bm{P}$.}  The algorithm above requires the matrix $\bm{D} = \bm{P}^T \bm{P}$ to have full rank, for a (near) optimum $\bm{C}^*$. If this does not hold, the idea will be to search over all the possible values of the rank. One natural idea is to solve the problem for all subsets of the columns $[s]$, but note that this takes time $\exp(s)$, which is much larger than our target running time $\exp(\text{poly}(k))$.  We thus perform a randomized procedure: for every guess $j$ for the column rank of $\bm{P}$, we take a random matrix $R \in \RR^{s \times j}$, and consider the regression problem
\begin{equation}\label{eq:regression-guess-j} \min_{\bm{C}, \bm{x}} \| (\bm{P} R) \bm{x} - \bm{b} \|^2 \le t. 
\end{equation}

Let $\bm{M}$ be the $\bm{P}$ matrix corresponding to an optimal solution $\bm{C}^*$, and suppose $j$ is its rank. In Lemma~\ref{lem:rand_rotation}, we show that if the entries of $R$ are drawn IID from $\cN(0,1)$, then with probability at least $3/4$, the matrix $\bm{M} R$ has full column rank. Thus, if we were to solve~\eqref{eq:regression-guess-j} with the optimal value of $t$ as our guess using the Solver discussed above, we would obtain an approximately optimal solution.

This leads to the following overall algorithm: 
\begin{itemize}
\item  For $j = 1, 2, \dots, s$:
\begin{itemize}
\item Sample a random matrix $R \in \RR^{s \times j}$ with entries drawn i.i.d. from $\cN(0,1)$.
\item Guess value of $t$ for the formulation~\eqref{eq:regression-guess-j} as discussed above.
\item If the determinant is not identically zero, call the core Solve subroutine and check obtained solution for feasibility.
\end{itemize}
\item Return the solution found as above with the least $t$.
\end{itemize}

Note that we can boost the success probability by sampling multiple $R$ for each guess of $j$.  
This completes the proof of our result, Theorem~\ref{thm:multiplicative}, modulo Lemma~\ref{lem:rand_rotation} which we prove below.

\begin{lemma}
\label{lem:rand_rotation}
    Let $M\in \RR^{d' \times s}$ be a matrix of rank $j$, and $R\in \RR^{s\times j}$ be a random matrix with entries drawn i.i.d. from $\mathcal{N}(0,1)$. Then the rank of $MR$ is equal to $j$ with probability $\ge 3/4$.
\end{lemma}
Note that since the ranks are equal, the column spans of $M$ and $MR$ must be the same.
\begin{proof}
We will prove a quantitative version of the statement. Suppose we denote the $j$th largest singular value of $M$ as $\tau = \sigma_j (M)$.  We will show that $\sigma_j (MR) \ge \frac{\sigma_j}{4 j^{3/2}}$ with probability $\ge 3/4$. 

Let $\cS$ be the span of the columns of $M$. For a random $x$ whose entries are drawn from $\cN(0,1)$, note that the distribution of $Mx$ is a (non-spherical) Gaussian on $\cS$. Moreover the covariance matrix has all eigenvalues $\ge \tau^2$. Let us now consider the matrix $MR$. We use a leave-one-out argument to show linear independence of its columns. I.e., we claim that every column of $MR$ has a projection of length at least $\frac{\tau}{4 j}$ orthogonal to the span of the other columns, with probability at least $3/4$. This implies (e.g., see~\cite{rudelson2009smallestsingularvaluerandom}) that $\sigma_j (MR) \ge  \frac{\sigma_j}{4 j^{3/2}}$ with probability at least $3/4$.

To see the claim, suppose we condition on all but the $i$th column of the matrix $R$, for some $1\le i\le j$. Then by the earlier observation, the $i$th column of $MR$  (denoted by $V_i$) is distributed as a non-spherical Gaussian on $\cS$ whose covariance matrix has all eigenvalues $\ge \tau^2$. Thus its projection to the space orthogonal to the span of the remaining columns of $MR$ (which are all fixed after conditioning) behaves as a Gaussian with at least one dimension and standard deviation $\ge \tau$. Thus by using the anti-concentration bound for a Gaussian (which states that the probability mass in any $\delta \tau$ sized interval is $\le \delta$), $V_i$'s projection orthogonal to the span of the other columns is at least $\frac{1}{4j} \cdot \tau$, with probability $1- \frac{1}{4 j}$. We can now take a union bound over all $1 \le i \le j$ to obtain a success probability of $3/4$. This completes the proof.
\end{proof}

\paragraph{Remark about precision.} The result above assumes that we use infinite precision for $R$. We now show how to avoid this assumption, with a slight loss in the parameters. 
Note that the key step in the above argument is showing that conditioned on the randomness in all but the $i$th column of $R$, the vector $V_i = \sum_{l=1}^s R_{il} M_l$ has a sufficiently large norm in the direction orthogonal to the span of the other columns in $MR$ (that are fixed due to the conditioning). For convenience, let $\Pi$ be the projector orthogonal to the span of the other columns, and denote $v_l = \Pi M_l$, and $X_{l} := R_{il}$. By the assumption on the least singular value, for any $\Pi$, we have that $\sum_l \norm{v_l}^2 \ge \tau^2$. This implies that there exists a coordinate $t \in [d']$ such that $\sum_l v_{lt}^2 \ge \frac{\tau^2}{d'}$. Just focusing on this coordinate, we could note that $\sum_l v_{lt} X_l$ is distributed as a Gaussian and thus conclude that the probability of the coordinate being $ < \frac{\tau}{4 j \sqrt{d'}}$ is $< 1/4j$. This leads to a slightly weaker (by a $\sqrt{d'}$ factor) bound on the least singular value, with the same success probability. However, this argument is more flexible, it lets us use ``discretized'' Gaussians.

\begin{lemma} For any $\delta, \eta \in (0, 1/2)$, there exists a centered distribution $\cY$ with the following properties:
\begin{enumerate}
\item The support of $\cY$ and the probability masses at each point in the support are all rational numbers of bit length $b = O(\log \frac{s}{\delta \eta} )$.
\item For all $\{a_i\}_{i=1}^s$ with $\sum_i a_i^2 = 1$, if $Y_1, Y_2, \dots, Y_s$ are independent random variables drawn from $\cY$, 
\[ \Pr[ |\sum_i a_i Y_i | < \delta ] \le 2\delta + \eta. \]
\end{enumerate}
\end{lemma} 
\begin{proof}
The proof goes by a discretization of the Gaussian distribution and a sequence of reductions. First, let $X_i$ be independent random variables distributed as $\cN(0,1)$.  Let $X_i'$ be independent random variables distributed as a \emph{truncated normal}, the distribution obtained from $\cN(0,1)$ by removing the mass at points $> \sqrt{\log( s/\eta)}$ and rescaling. We begin by noting that
\begin{equation}\label{eq:prob-bound-xxprime}
\Pr[ | \sum_i  a_i X_i' | < \delta ] \le \left( 1+\eta  \right) \Pr[ |\sum_i a_i X_i | < \delta ].
\end{equation}
To see this, let us write $\bm{X}$ to be the vector $(X_1, X_2, \dots, X_s)$. If $f(\bm{X})$ and $f'(\bm{X})$ denote the probability density functions of the multi-dimensional normal and the truncation version respectively, by the choice of the truncation, we have, for all $\bm{X}$,
\[ f' (\bm{X} ) \le \left( 1+\frac{\eta}{s} \right)^s f( \bm{X} ) \le (1+\eta) f(\bm{X}). \]
Furthermore, if any of the coordinates of $\bm{X}$ is $> \sqrt{\log (s/\eta) }$, we have $f' (\bm{X} ) = 0$. 
Next, note that the LHS of~\eqref{eq:prob-bound-xxprime} can be written out as an integral, and every term also appears on the probability on the right, albeit with the measure $f'$ replaced with $f$. Thus, using the bound above,~\eqref{eq:prob-bound-xxprime} follows.

Next, we discretize the interval $[ -\lceil  \sqrt{\log (s/\eta)} \rceil, \lceil  \sqrt{\log (s/\eta)} \rceil ] $ into integral multiples of $1/M$, where $M$ is a parameter we will choose later. To the point $i/M$, we assign the mass that the truncated Gaussian assigns to the interval $[i - \frac{1}{2M}, i+\frac{1}{2M}] $. We call this discrete distribution $\cD$.  Let $Y_1, Y_2, \dots, Y_s$ be IID samples from $\cD$. We can write $Y_i = X_i' + Z_i$, where $X_i'$ is drawn from the truncated Gaussian, and $|Z_i| \le \frac{1}{2M}$. Thus, for $M > \frac{s}{\delta}$, we claim that
\[ \Pr[ |\sum_i a_i Y_i| < \delta] \le \Pr[ | \sum_i a_i X_i' | <  2\delta ]. \]
This follows because $|\sum_{i \in [s]} a_i Z_i| \le \frac{s}{2M} < \delta$.  Using the earlier bounds, this implies that
\[ \Pr[ |\sum_i a_i Y_i| < \delta] \le (1+\eta) \cdot 2\delta. \]
This almost completes the proof, because we have obtained a discrete distribution with the desired anti-concentration bound. But as such, note that the probability values can require very high precision. This turns out to be easy to correct: we can take $\cY$ to be any distribution on the same support as $\cD$ with $d_{TV} (\cY, \cD) < \epsilon$, and we can conclude (using a coupling argument and a union bound), that if $W_1, W_2, \dots, W_s$ are drawn IID from $\cY$,
\[ \Pr[ |\sum_i a_i W_i| < \delta] \le \Pr[ | \sum_i a_i Y_i | <  \delta ] + \epsilon s. \]
To complete the argument, we need to ensure that $\epsilon < \frac{\eta}{s}$; this can be achieved using probability values with bit complexity only $O(\log (s/\eta))$, thus completing the proof.
\end{proof}

\subsection{Projective Non-negative Matrix Factorization}
In projective non-negative matrix factorization, the basis matrix $\bm{U}\in \RR^{d\times k}$ is constrained to have non negative entries. More formally, the mathematical program formulation for Projective Non-negative Matrix Factorization (NMF) is
\begin{align}
\label{prog: NMF}\min&: \|\bm{A}-\bm{U}\bm{U}^T\bm{A}\|_F^2 \tag{NMF}\\
    \bm{U}^T\bm{U}&=\bm{I}_k, \; \bm{U}\in \RR^{d\times k}_{\geq 0}.
\end{align}
There is a alternate formulation of \ref{prog: NMF} that better aligns with the name of the problem and is well suited to apply our framework:
\begin{align}
    \label{prog: NMF_original} \min&: \|\bm{A}-\bm{W}\bm{H}\|_F^2 \tag{NMF-alternate}\\
    \bm{W}\in \RR^{d\times k}_{\geq 0}& \textnormal{ has orthogonal columns }.
\end{align}
\begin{lemma}
    The programs \ref{prog: NMF} and \ref{prog: NMF_original} are equivalent.
\end{lemma}
\begin{proof}
    Using \Cref{lem:csa-alt-equiv}, we know that \ref{prog: NMF} is equivalent to 
    \begin{align*}
        \min&: \|\bm{A}-\bm{U}\bm{H}\|_F^2 \\
    \bm{U}^T\bm{U}&=\bm{I}_k, \; \bm{U}\in \RR^{d\times k}_{\geq 0}.
\intertext{
    Since $\bm{H}$ is unconstrained, it can absorb the normalization of the columns of $\bm{U}$. This gives 
}
       \min&: \|\bm{A}-\bm{W}\bm{H}\|_F^2 \\
   \bm{W}\in \RR^{d\times k}_{\geq 0} & \textnormal{ has orthogonal columns }
    \end{align*}
    which is exactly \Cref{prog: NMF_original}.
\end{proof}
\begin{lemma}
\label{lem:NMF-W-characterization}
The set of matrices 
\begin{align}
    \cW&:=\{\bm{W}\in \RR^{d\times k}_{\geq 0}:  \bm{W} \textnormal{ has orthogonal columns }\}
\intertext{ is equal to the set }
     \overline{\cW}&:= \{\bm{W}\in \RR^{d\times k}_{\geq 0}: \|\bm{W}_{i,.}\|_0\leq 1 \quad \forall i\in [d]\}.
 \end{align}
\end{lemma}
\begin{proof}
    For any $\bm{W}\in \cW$, if there exists a row $i\in [d]$ and two distinct indices $j,j' \in [k]$ such that $\bm{W}_{i,j},\bm{W}_{i,j'}\neq 0$, then by the non-negativity constraint, these non zero values are in fact strictly positive. The dot product of the columns $\bm{W}_{.,j}$ and $\bm{W}_{.,j'}$is at least $\bm{W}_{i,j}\cdot\bm{W}_{i,j'}>0$ which contradicts the orthogonality of the columns of $\bm{W}$. This implies that there is at most one non-zero entry in every row of $\bm{W}$ which further implies that $\bm{W}\in \overline{\cW}$. 

    For any $\bm{W}\in \overline{\cW}$, the orthogonality of the columns is straight forward because for any two distinct indices $j,j'\in [k]$, the dot product of the columns $\bm{W}_{.,j}$ and $\bm{W}_{.,j'}$ is zero because either of $\bm{W}_{i,j}, \bm{W}_{i,j'}$ is equal to zero for every $i\in [d]$.
\end{proof}

\begin{lemma}
\label{lem:NMF-regression}
    For any given $\bm{B}\in \RR^{d\times r}$ and $\bm{H}\in \RR^{k\times r}$, we can solve the program 
    \begin{align*}
        \min&: \|\bm{B}-\bm{W}\bm{H}\|_F^2 \\
   \bm{W}\in \RR^{d\times k}_{\geq 0} & \textnormal{ has orthogonal columns }
    \end{align*}
     exactly in time $O(dkr)$.
\end{lemma}
\begin{proof}
    Using \Cref{lem:NMF-W-characterization}, we can re-write the optimization problem as 
    \begin{align*}
         \min&: \|\bm{B}-\bm{W}\bm{H}\|_F^2 \\
   \|\bm{W}\|_{0} & \leq 1 \quad \forall i\in [d],\; \bm{W}\in \RR^{d\times k}_{\geq 0}. 
    \end{align*}
    Since the rows of $\bm{W}$ are independent variables, we can decompose the problem into 
    \begin{align*}
        \sumL_{i=1}^d \min_{ w_i\in \RR^{k}_{\geq 0},\,\|w_i\|_0\leq 1}: \|b_i-\bm{H}^T w_i\|_2^2
    \end{align*}
    where $b_i,w_i$ are the $i^{th}$ columns of $\bm{B}^T$ and $\bm{W}^T$ respectively. Each problem $\min_{ w_i\in \RR^{k}_{\geq 0},\,\|w_i\|_0\leq 1}: \|b_i-\bm{H}^T w_i\|_2^2$ can be solved by looking at the $k$ cases for the non-zero entry if $w_i$. For every choice of non-zero entry, say $j\in [k]$, the resulting minimization problem is $min_{\lambda \geq 0}: \|b_i - \lambda h_{j}\|_2^2$ where $h_j$ is the $j^{th}$ column of $\bm{H}^T$. The optimal choice of $\lambda$ is $\max(0,\langle b_i,h_j\rangle/\|h_j\|_2^2)$. The only computation we had to do is to evaluate the dot products between $b_i,h_j$ for $i\in [d], j\in [k]$ which takes $O(dkr)$ time. 
\end{proof}

\begin{theorem}[Additive approximation for NMF] \label{thm:NMF-additive}
Given an instance $\bm{A}\in \RR^{d\times n}$ of Non-negative matrix factorization, there is an algorithm that computes a $\bm{U}\in \RR^{d\times k}_{\geq 0}, \; \bm{U}^T\bm{U}=\bm{I}_k$ such that 
\begin{align*}
    \|\bm{A}-\bm{U}\bm{U}^T\bm{A}\|_F^2 \leq (1+\varepsilon)\cdot \opt + O(\delta\cdot \|\bm{A}\|_F^2)
\end{align*}
in time $O(dk^2/\varepsilon)\cdot (1/
\delta)^{O(k^2/\varepsilon)}$. For any $0<\delta<1$. 
\end{theorem}
\begin{proof}
   Let $\bm{U}^*$ be the optimal solution to \ref{prog: NMF}. This implies that $\bm{H}={\bm{U}^T}^*\bm{A}$ is an optimal solution to \ref{prog: NMF_original}. We first replace the instance with a smaller instance $\bm{B}$. Then we search for every row of $\bm{H}$ exactly as in the proof of \Cref{lem:structural-CSA} to obtain a solution $\bm{U}\in \RR^{d\times k}_{\geq 0}, \; \bm{U}^T\bm{U}=\bm{I}_k$ such that 
   \begin{align*}
    \|\bm{A}-\bm{U}\bm{U}^T\bm{A}\|_F^2 \leq (1+\varepsilon)\cdot \opt + O(\delta\cdot \|\bm{A}\|_F^2)
\end{align*}
in time $T_0+T_1\cdot(1/\delta)^{O(rk)}$. Where $T_0$ is the time required to obtain $\bm{B}$ and $T_1$ is the time required to solve for the optimal $\bm{W}$ in the program 
    \begin{align*}
        \min&: \|\bm{B}-\bm{W}\bm{H}\|_F^2 \\
   \bm{W}\in \RR^{d\times k}_{\geq 0} & \textnormal{ has orthogonal columns }
    \end{align*}
    We know that $T_1= O(dkr)$ using \Cref{lem:NMF-regression} and $T_0=O(nrd^2\cdot H)$ from \Cref{lem:svd-time}. We hide $T_0$ as it is negligible. 
\end{proof}

\begin{theorem}[Multiplicative approximation for NMF]\label{thm:NMF-multiplicative}
Given an instance $\bm{A}\in \RR^{d\times n}$ of Non-negative matrix factorization with integer entries of absolute value at most $\gamma$ in $\bm{A}$, there is an algorithm that computes a $\bm{U}\in \RR^{d\times k}_{\geq 0}, \; \bm{U}^T\bm{U}=\bm{I}_k$ such that 
\begin{align*}
    \|\bm{A}-\bm{U}\bm{U}^T\bm{A}\|_F^2 \leq (1+\varepsilon)\cdot \opt 
\end{align*}
in time $(nd\gamma/\varepsilon)^{O(k^3/\varepsilon)}$. 
\end{theorem}
\begin{proof}
Using \Cref{lemma:subspace_error_lb}, we know that $\|\bm{A}\|_F^2/\|\bm{A}-\bm{A}_k\|_F^2\leq (nd\gamma)^{O(k)}$. Setting $\delta= \varepsilon\|\bm{A}-\bm{A}_k\|_F^2/\|\bm{A}\|_F^2\geq \varepsilon (nd\gamma)^{-O(k)}$ in \Cref{cor:CSE-additive} gives the desired time complexity.  
\end{proof}

\subsection{$k$-means clustering \cite{cohen2015dimensionality}}\label{sec:k-means}
In the $k$-means problem, we are given a collection of data points $a_1,\dots,a_n \in \RR^d$. The objective is to find $k$-centers $c_1,\dots,c_k \in \RR^d$ and an assignment $\pi: [n]\rightarrow [k]$ that minimizes:
\begin{align}
\label{prog:k-means-geo}    \sumL_{i=1}^n \|a_i-c_{\pi(i)}\|_2^2 \tag{$k$-means}.
\end{align}
Let $\bm{A}\in \RR^{n\times d}$ and  $\bm{C} \in \RR^{k\times d}$ matrices with $a_i$ and $c_i$ as their $i^{\textnormal{th}}$ rows respectively (note that this differs from the notation we used for previous applications). Let $\bm{\Pi}\in \RR^{n\times k}$ be the matrix such that $\bm{\Pi}_{i,j}=\mathbbm{1}[\pi(i)=j]$. Using this notation, the \ref{prog:k-means-geo} problem can be written as 
\begin{align}
    \label{prog:k-means-fac} \min&: \|\bm{A}- \bm{\Pi}\bm{C}\|_F^2 \tag{$k$-means-matrix}\\
    &\textnormal{Each row of } \bm{\Pi} \textnormal{ is a standard basis vector}. \nonumber
\end{align}
 Observe that \ref{prog:k-means-fac} is a special case of \ref{prog: NMF_original} where $\bm{W}$ is additionally constrained to have all non-zero entries to be equal to $1$. Also, the \ref{prog:k-means-geo} and \ref{prog:k-means-fac} correspond to the \ref{prog:CSA-geo} and \ref{prog:CSA-factorization} formulations of the same problem. The corresponding \ref{prog:CSA} version is 
\begin{align*}
 \label{prog:k-means-csa} \min&: \|\bm{A}-\bm{U}\bm{U}^T\bm{A}\|_F^2 \tag{$k$-means-CSA}\\
 \bm{U}_{i,j}&=1/\sqrt{\|U_{.,j}\|_0} \quad \forall i\in [n],\, j\in [k].  
\end{align*}
The three main steps in our algorithm are:
\begin{enumerate}
    \item \textbf{Reduction:} The first step is to reduce the number of rows and columns of the target matrix $\bm{A}$. 
    \begin{enumerate}
        \item \textbf{Columns:} Replace the matrix $\bm{A}$ with the matrix $\bm{B}$ as in \Cref{lem:coreset_p=2}. This reduces the number of columns (dimension of the data-points) to $r=k+k/\varepsilon$. 
        \item \textbf{Rows:} For any fixed set of centers (selected from the rows of) $\bm{C}$, the cost induced by the centers is defined as
\begin{align*}
  \textnormal{Cost}(\bm{A}, \bm{C}):= \sumL_{i=1}^n \textnormal{dist}(a_i,\bm{C})^2.
\end{align*}
Where $\textnormal{dist}(a,\bm{C}):= \min_{c\in \bm{C}}\|a-c\|_2$. A strong coreset for the $k$-means instance $\bm{A}$ is a subset $S\subseteq [n]$ of indices and weights $w_i$ corresponding to each index $i\in S$ such that for any set of centers $\bm{C}$, we have 
\begin{align*}
   \textnormal{Cost}_{w,S}(\bm{A}, \bm{C}):=\sumL_{i\in S}w_i \textnormal{dist}(a_i, \bm{C}) \in (1\pm \varepsilon)\cdot \textnormal{Cost}(\bm{A}, \bm{C}).
\end{align*}
Coresets for $k$-means of optimal size $\widetilde{O}(k\varepsilon^{-2}\min\{\sqrt{k}, \varepsilon^{-2}\})$ are known (See \cite{cohenkmeanscoreset2022} for upper-bound and \cite{Huang2024coresetoptimallb} for matching lower-bound ). Any algorithm that efficiently computes a coreset of size $q=\poly(k/\varepsilon)$ can be used as a black box for our purposes. After using such a coreset, the new formulation is 
    \begin{align}
\label{prog:k-means-reduced}     \min&: \|\bm{B}- \bm{\Pi}\bm{C}\|_F^2 \tag{k-means-reduced}\\
 \label{eqn:row-condition-k-means-reduced}   &\textnormal{Each row of } \bm{\Pi} \textnormal{ has exactly one non-zero entry equal to } w_i. 
\end{align}
where $\bm{B}\in \RR^{q\times r}, \; \bm{\Pi}\in \RR^{q\times k}$ with their rows indexed by the set $S$ defined by the coreset and the rows of $\bm{B}$ are the scaled rows of $\bm{A}$ according to the weight defined by the coreset for that row.  
    \end{enumerate} 

\item \textbf{Enumeration:} A naive approach is to simply enumerate all the $k^q$ possible $\bm{\Pi}$ matrices by choosing the position of the non-zero element in each row. Simply put, we go through all possible $k$-clusterings of the coreset elements. Optimal choice of centers can be computed as the weighted mean of the coreset elements in each cluster. This allows us to identify the optimal $\bm{\Pi}$.

Let $\bm{\Pi}^*$ be the optimal choice of $\bm{\Pi}$ in the \ref{prog:k-means-reduced} program. Using \Cref{lem:kmeans-leverage-score} and \Cref{lem:k-means-enum-smart}, we enumerate over the $O(\log n \cdot k \cdot \poly(k/\varepsilon))^{O(k\log k +k/\varepsilon)}=O(k/\varepsilon \cdot \log n )^{\widetilde{O}(k/\varepsilon)}$ number of possible pairs of matrices $\bm{S}\bm{B}$ and $\bm{S}\bm{\Pi}$. For each such pair, we find the $\bm{C}$ that minimizes $\|\bm{S}\bm{B}- \bm{S}\bm{\Pi}\bm{C}\|_F^2$. For every such $\bm{C}$, evaluate the cost induced by these centers with the coreset $(w, S)$. Let $\overline{\bm{C}}$ be the set of centers that has the lowest cost with respect to the coreset from the enumeration described before. The cost induced by this set of centers is  
\begin{align}
  \min_{\bm{\Pi}}\|\bm{B}-\bm{\Pi}\overline{\bm{C}}\|_F^2  &\leq   \min_{\bm{\Pi}}\|\bm{B}-\bm{\Pi}\widehat{\bm{C}}\|_F^2 \\
  & \leq \|\bm{B}-\bm{\Pi}^*\widehat{\bm{C}}\|_F^2 \\
  & \leq (1+\varepsilon)\|\bm{B}-\bm{\Pi}^*\bm{C}^*\|_F^2 \\
  & \leq (1+\varepsilon)^2 \cdot \opt.
  \intertext{Using the coreset property, we imply that the cost of the centers $\overline{C}$ on the original instance $\bm{A}$ is at most $(1+\varepsilon)^3\cdot \opt$.}
\end{align}

\end{enumerate}

We start with the following known result (see Theorem 38 of~\cite{clarkson2017low}).
\begin{lemma}
\label{lem:kmeans-leverage-score}
Given matrices $\bm{B}\in \RR^{q\times r}$ and $\bm{\Pi}^* \in \RR^{q\times k}$, there exists a matrix $\bm{S}\in \RR^{t\times q}$ and such that 
    \begin{enumerate}
   \item  \label{item:kmeans-leverage-score-1} Each row of $\bm{S}$ contains exactly one positive non-zero element from the set $W=\{2^{i}: 0\leq i\leq N\}$.
    \item  \label{item:kmeans-leverage-score-2} If $\bm{C}^*=\argmin_{\bm{C}\in \RR^{k\times r}}\|\bm{B}-\bm{\Pi}^*\bm{C}\|_F^2$ and $\widehat{\bm{C}}=\argmin_{\bm{C}\in \RR^{k\times r}}\|\bm{S}\bm{B}-\bm{S}\bm{\Pi}^*\bm{C}\|_F^2$, then 
    \begin{align}
        \|\bm{B}-\bm{\Pi}^*\widehat{\bm{C}}\|_F^2\leq (1+\varepsilon)\cdot \|\bm{B}-\bm{\Pi}^*\bm{C}^*\|_F^2.
    \end{align}
    \item $t=O(k\log k +k/\varepsilon)$.
    \end{enumerate}
\end{lemma}
Note that~\cite{clarkson2017low} do not require the non-zero element of $\bm{S}$ to come from $W$. Indeed, it will be proportional to the leverage score. However, note that we can ``discretize'' the leverage scores (while keeping a factor two approximation to each one), and still obtain all the guarantees that we require. Finally, since the leverage scores add up to the matrix dimension, we have the bound $N = O(\log n)$.

\begin{lemma}
    \label{lem:k-means-enum-smart}
   Given a matrix $\bm{B}\in \RR^{q\times r}$, the number of possible matrices
   \begin{enumerate}
       \item of the form $\bm{S}\bm{B}$ is at most $O(Nq)^t$.
       \item of the form $\bm{S}\bm{\Pi}$ where $\bm{\Pi}$ satisfies \Cref{eqn:row-condition-k-means-reduced} is at most $O(Nkq)^t$.
   \end{enumerate}
where $\bm{S}$ satisfies property \ref{item:kmeans-leverage-score-1} in \Cref{lem:kmeans-leverage-score}.
\end{lemma}
\begin{proof}
    Each row of $\bm{S}\bm{B}$ is simply a row of $\bm{B}$ that is scaled by $2^{i}$ for some $0\leq i\leq N$. This leaves $Nq$ choices for each of the $t$ rows of $\bm{S}\bm{B}$ which is $((N+1)q)^t$ possibilities. Each row of $\bm{S}\bm{\Pi}$ is a row of $\bm{\Pi}$ scaled by $2^{i}$ for some $0\leq i \leq N$. The choices to make for each row of $\bm{S}\bm{\Pi}$ is a row of $\bm{\Pi}$ (which includes choices for non-zero element and weight $w_j$ for $j\in [q]$) and a scaling factor from $\bm{S}$. This leaves $(N+1)kq$ choices for each of the $t$ rows of $\bm{S}\bm{\Pi}$. 
\end{proof}

\begin{theorem}
    \label{thm:k-means-runtime}
    Given an instance $\bm{A}\in \RR^{n\times d}$ of $k$-means, there is an algorithm that computes a $(1+\varepsilon)$-approximate solution to \ref{prog:k-means-geo} in $O(nnz(\bm{A})+ 2^{\widetilde{O}(k/\varepsilon)}+n^{o(1)})$ time. 
\end{theorem}
\begin{proof}
    The time complexity of the three step procedure is dominated by the enumeration step which takes time $O(\log n \cdot k/\varepsilon)^{\widetilde{O}(k/\varepsilon)}$ time. If $\log n \leq (k/\varepsilon)^2$, then this running time is $O(k/\varepsilon)^{\widetilde{O}(k/\varepsilon)}=2^{\widetilde{O}(k/\varepsilon)}$. Otherwise, if $\log n \geq (k/\varepsilon)^2$, then the running time is $(\log n)^{\widetilde{O}(\sqrt{\log n})} =n^{o(1)}$. 
\end{proof}

\subsection{Sparse-PCA \cite{del2022sparse}}
The sparse PCA problem is a well-studied variant of PCA in which the components found are required to be sparse. In other words, the basis matrix $\bm{U}\in \RR^{d\times k}$ is constrained to have sparsity requirements. There are two natural ways to formalize this question: the first is by requiring $\bm{U}$ to have at most $s$ non-zero entries in total. Another is to require the number of non-zero {\em rows} of $\bm{U}$ to be bounded by a parameter $s$. In the popular case of $d=1$, both of these definitions coincide. Let us focus on the first variant for now.\footnote{Our result follows via a black-box application of algorithms from~\cite{del2022sparse}; since their algorithms work for both variants, so do our results.} More formally, the mathematical program formulation we consider is
\begin{align}
    \label{prog:sparse-PCA-max} \max&: \langle \bm{A}\bm{A}^T, \bm{U}\bm{U}^T \rangle \tag{sparse-PCA-max}\\
    \bm{U}^T\bm{U}&=\bm{I}_k, \; \sum_{j\in [k]}\|\bm{U}_{.,j}\|_0 \leq s.
\end{align}
Program \ref{prog:sparse-PCA-max} can also be formulated as a minimization version 
\begin{align}
    \label{prog:sparse-PCA-min} \min&: \|\bm{A}-\bm{U}\bm{U}^T\bm{A}\|_F^2 \tag{sparse-PCA-min}\\
    \bm{U}^T\bm{U}&=\bm{I}_k, \; \sum_{j\in [k]}\|\bm{U}_{.,j}\|_0 \leq s.
\end{align}

The following theorem from \cite{del2022sparse} shows how to find an optimal solution to \ref{prog:sparse-PCA-max} (and hence also to \ref{prog:sparse-PCA-min}) when $\textnormal{rank}(\bm{A}\bm{A}^T)=\textnormal{rank}(\bm{A})=t$.
\begin{theorem}[Theorem 1 in \cite{del2022sparse}]
\label{thm:sparsepca}
    There is an algorithm that finds an optimal solution to \ref{prog:sparse-PCA-max} in 
    \begin{align}
        O\left(d^{\min\{k,t\}(t^2+t)}\left(\min\{k,t\}dt^2+ d\log d\right)\right),
    \end{align}
    where $t$ denotes the rank of the matrix $\bm{A}$.
\end{theorem}

\begin{theorem}\label{thm:sparse-PCA}
    Given an instance $(\bm{A}\in \RR^{d\times n},k,s)$ of sparse-PCA, there is an algorithm that runs in time 
    \begin{align}
        O\left(d^{kr^2+kr}\left(dkr^2+ d\log d\right)\right)
    \end{align}
with $r=k+k/\varepsilon$ that computes a $\varepsilon \|\bm{A}-\bm{A}_k\|_F^2$ additive approximate solution to both \ref{prog:sparse-PCA-max} and \ref{prog:sparse-PCA-min}. This is guarantees as a $(1+\varepsilon)$-approximate solution to \ref{prog:sparse-PCA-min} because $\|\bm{A}-\bm{A}_k\|_F^2$ is a lower bound to \ref{prog:sparse-PCA-min}. 
\end{theorem}
\begin{proof}
 First step is to replace $\bm{A}$ with the matrix $\bm{B}$ as in \Cref{lem:coreset_p=2}. This step takes time $T_0$. Solve for \Cref{prog:sparse-PCA-max} exactly using \Cref{thm:sparsepca}, with $t=k+k/\varepsilon$. This step takes time   $ O\left(d^{kr^2+kr}\left(dkr^2+ d\log d\right)\right)$. 
 Using \Cref{lem:approximate_soln}, the solution obtained is a $(1+\varepsilon)$-approximate solution to \ref{prog:sparse-PCA-min}. In fact, for $p=2$, we know that the error is at most $\varepsilon \|\bm{A}-\bm{A}_k\|_F^2$ which implies that this also gives an additive approximation of $\varepsilon \|\bm{A}-\bm{A}_k\|_F^2$ to both the minimization and maximization versions. Because the objective for the maximization is the negative of the minimization objective added with $\|\bm{A}\|_F^2$.
\end{proof}

\section{Hardness of Column Subset Selection with Partition Constraint}\label{sec:css-hardness}

In this section, we show that \fairCSS{} is at least as hard as the well-studied \emph{sparse regression} problem~\cite{natarajan1995sparse,FKT15,har2018approximate,gupte2021fine}. In particular, our hardness implies that \fairCSSx{} remains hard even if the number of groups is only two, or if we allow violating the given partition capacity constraints by a logarithmic factor. First, we define the \fairCSSx{} problem and the sparse regression problem formally. 

\begin{defn}[Column Subset Selection with a Partition Constraint]
\label{def:fair-css}
    In an instance of the \emph{\fairCSSx} (\emph{\fairCSS}) problem, we are 
    given a matrix $\bm{A}\in \RR^{m\times n}$ and a partition matroid $\cP=([n]=P_1\uplus\dots \uplus P_{\ell},\;\cI), \; \cI=\{S\subseteq [n]: |S\cap P_{t}|\leq k_{t}, \; \forall t \in [\ell]\}$ defined on the set of column indices $[n]$. 
    
    The objective is to select a (index) subset $S\in \cI$ of columns of $\bm{A}$ in order to minimize the squared projection cost of all the column vectors of $\bm{A}$ onto the span of the column space induced by the subset of columns corresponding to $S$
    \begin{align}
        \text{cost}_{S}(\bm{A}):= \sumL_{i\in [n]}\|\text{proj}_{\text{span}^\bot(S)}(a_i)\|_2^2 = \|\bm{A}-\bm{A}_S\bm{A}_S^+\bm{A}\|_F^2,
\end{align}  
  where $a_i$ is the column vector corresponding to column index $i$ in $\bm{A}$ and $\bm{A}_S$ is the matrix with columns from $\bm{A} $  corresponding to $S$.
\end{defn}


\begin{defn}[$(g,h)$-Sparse Regression]
Given a matrix $\bm{B} \in \mathbb{R}^{m\times n}$ and a positive integer $s$, for which there exists an unknown vector $x^* \in \mathbb{R}^n$ such that $\|x^*\|_0\leq s$ and $\bm{B}x^*= \mathbbm{1}$, the goal is to output an $x \in \mathbb{R}^n$ with $\|x\|_0 \leq s\cdot g(n)$ such that $\|\bm{B}x-\mathbbm{1}\|_2^2\leq h(m,n)$. 
\end{defn}


The sparse regression problem is known to be computationally hard. In particular, 
\begin{theorem}[\cite{FKT15}]
    Let $0<\delta <1$. If there is a deterministic polynomial time algorithm $\cA$ for $(g,h)$-sparse regression, for which $g(n)=(1-\delta)\ln n$ and $h(m,n)=m^{1-\delta}$, then $\mathrm{SAT} \in \mathrm{DTIME}(n^{O(\log \log n)})$.   
\end{theorem}

Next, we prove our main hardness of approximation result for \fairCSSx{}.

\begin{theorem}\label{thm:hard-CSS}
Assuming $\mathrm{SAT} \notin \mathrm{DTIME}(n^{O(\log \log n)})$, the \fairCSSx{} problem is hard to approximate to any multiplicative factor $f$, even in the following special cases:

(i) The case of $\ell=2$ groups, where the capacities on all the groups are the same parameter $s$.

(ii) The case where the capacities on all the groups are the same parameter $s$, and we allow a solution to violate the capacity by a factor $g(n) = o(\log n)$, where $n$ is the total number of columns in the instance.
\end{theorem}
\begin{proof}
The proof is via a reduction from sparse regression. First, we show a hardness for two groups (part (a) of the Theorem). Consider an instance of sparse regression, given by an $m\times n$ matrix $B$ and parameter $s$. Now consider a matrix $A$ whose columns are $A_1 \cup A_2$, defined as follows. $A_1$ is an $(m+s)\times n$ matrix whose $i$th column is the $i$th column of $B$ appended with $s$ zeros. $A_2$ is an $(m+s) \times (s+1)$ matrix whose columns we denote by $u_1, u_2, \dots, u_{s+1}$. We set $u_i = C \cdot e_{m+i}$ for $1 \le i \le s$, and $u_{s+1} = D \cdot (\mathbf{1} \oplus \mathbf{0}_s)$, for appropriately chosen parameters $C > D$.\footnote{As is standard, $\mathbf{1} \oplus \mathbf{0}_s$ is simply the all ones vector (here in $m$ dimensions) with $s$ zeros appended.}

Consider any solution that chooses exactly $s$ columns from $A_1$ and $A_2$. In the YES case of sparse regression, where there exists an $s$-sparse $x^*$ with $Bx^* = \mathbf{1}$, by choosing the columns corresponding to the support of  $x^*$ from $A_1$ along with the columns $u_1, \dots, u_s$ from $A_2$, we obtain an approximation error at most $\| A_1 \|_F^2$. Consider the NO case of sparse regression. We will choose $C$ large enough, so that even if one of the $u_i$ for $i\le s$ is not chosen, the error is $\ge C$. Next, suppose all the $\{u_i\}_{i \in [s]}$ are chosen. For any choice of $s$ columns from $A_1$, the error on the column $u_{s+1}$ is at least $D \cdot h(m,n)$, by assumption. Thus in either case, the approximation error is $\ge \min \left( C, D \cdot h(m,n) \right)$. We can now choose $C, D$ large enough (e.g., $> f \cdot n \| A_1 \|_F^2$), and obtain the desired hardness of approximation.

Next, suppose we are allowed to choose $\alpha s$ columns from each group, for some slack parameter $\alpha$ (assumed to be an integer $\ge 1$ and $< g(n)$, where the latter function comes from the hardness for sparse regression). Now let $T$ be a parameter we will choose later (integer $\ge 1$), and consider an instance of \fairCSSx{} where we have $(T + 1)$ groups of vectors (matrices), $A_1, A_2, \dots, A_{T+1}$, and the vectors (columns) have dimension $T m$. We view each column vector as consisting of $T$ blocks of size $m$. For $1\le j \le T$, the columns of $A_j$ are identical to those of $B$ in the $j$th block, and zero everywhere else. The matrix $A_{T+1}$ has $T$ columns, denoted $u_1, \dots, u_T$, where $u_{j}$ is the vector that has $\mathbf{1}$ in the $j$th block and zero everywhere else, scaled by parameter $D$.

As before, in the YES case of sparse regression, the approximation error is $\le T \| A_1\|_F^2$. In the NO case, consider any solution that chooses at most $\alpha s$ vectors from each $A_j$. By assumption, the error in the $j$th block (of $u_{j}$) is at least $h(m,n)$, for any vector $u_{j}$ that is not picked from $A_{T+1}$. If we set $T > 2\alpha s$, then at least $(T/2)$ of the vectors $u_{j}$ cannot be picked, and so the total error is at least $D \cdot (T/2) h(m, n)$. Again, we can choose $D$ large enough to obtain the desired hardness.
\end{proof}
This strong hardness of approximation further motivates the study of a relaxed variant, in which the set of vectors in the small-size summary $S$, rather than being a subset of $\bm{A}^{(1)}, \cdots, \bm{A}^{(\ell)}$, are instead required to belong to the {\em subspaces} spanned by the columns in each group $\bm{A}^{(1)}, \cdots, \bm{A}^{(\ell)}$. This is precisely our \fairSA{}problem.

{\bf Acknowledgments:} Aditya Bhaskara was supported by NSF CCF-2047288. David P. Woodruff was supported by a Simons Investigator Award and Office of Naval Research award number N000142112647.

\bibliographystyle{alpha}
\bibliography{bibdb}

\end{document}